\documentclass[onefignum,onetabnum]{siamart250211}

\usepackage{lipsum}
\usepackage{amsfonts}
\usepackage{graphicx}
\usepackage{epstopdf}
\ifpdf
  \DeclareGraphicsExtensions{.eps,.pdf,.png,.jpg}
\else
  \DeclareGraphicsExtensions{.eps}
\fi

\usepackage{mathrsfs}
\usepackage{amsmath}
\usepackage{indentfirst}
\usepackage{amssymb}
\usepackage{bm}
\usepackage{subfig}
\usepackage{multirow}
\usepackage{array}
\usepackage{graphicx}
\usepackage{makecell}
\usepackage{enumerate}

\usepackage{algorithm}
\usepackage{algpseudocode} 

\DeclareMathOperator{\supp}{supp}
\DeclareMathOperator{\dist}{dist}

\algrenewcommand\algorithmicrequire{\textbf{Input:}}
\algrenewcommand\algorithmicensure{\textbf{Output:}}

\def\XXint#1#2#3{{\setbox0=\hbox{$#1{#2#3}{\int}$ }
\vcenter{\hbox{$#2#3$ }}\kern-.6\wd0}}

\newsiamremark{remark}{Remark}
\newsiamremark{assumption}{Assumption}
\newsiamremark{hypothesis}{Hypothesis}
\crefname{hypothesis}{Hypothesis}{Hypotheses}
\newsiamthm{claim}{Claim}
\newsiamremark{fact}{Fact}
\crefname{fact}{Fact}{Facts}

\headers{Higher-Order Barren Plateau}{Y. Huang and Y. Wang}

\title{Analysis of Hessian Scaling for Local and Global Costs in Variational Quantum Algorithms \thanks{Submitted to the editors DATE.  
\funding{}}}

\author{Yihan Huang\thanks{Department of Mathematics, Faculty of Science, National University of Singapore, 10 Lower Kent Ridge Road, Singapore.}
\and Yangshuai Wang\thanks{Corresponding author. Department of Mathematics, Faculty of Science, National University of Singapore, 10 Lower Kent Ridge Road, Singapore
  (\email{yswang@nus.edu.sg}).}}

\theoremstyle{plain}

\theoremstyle{definition}

\DeclareMathOperator{\Var}{Var}
\DeclareMathOperator{\Cov}{Cov}
\DeclareMathOperator{\Tr}{Tr}
\newcommand{\E}{\mathbb{E}}
\newcommand{\R}{\mathbb{R}}
\newcommand{\Z}{\mathbb{Z}}
\newcommand{\bth}{\boldsymbol{\theta}}
\newcommand{\ej}{\mathbf{e}_j}
\newcommand{\ek}{\mathbf{e}_k}

\newcommand{\Hjj}{H_{jj}}
\newcommand{\Hjk}{H_{jk}}

\usepackage{amsopn}

\ifpdf
\hypersetup{
  pdftitle={Analysis of Hessian Variance Scaling for Local and Global Cost Functions in Variational Quantum Algorithms},
  pdfauthor={[Author One], [Author Two], [Author Three]}
}
\fi

\usepackage{xcolor}

\theoremstyle{plain}

\theoremstyle{definition}

\definecolor{yscol}{HTML}{6622AA}

\begin{document}

\maketitle

\begin{abstract}
Barren plateaus in variational quantum algorithms are typically described by gradient concentration at random initialization. In contrast, rigorous results for the Hessian, even at the level of entry-wise variance, remain limited. In this work, we analyze the scaling of Hessian-entry variances at initialization. Using exact second-order parameter-shift identities, we write $H_{jk}$ as a constant-size linear combination of shifted cost evaluations, which reduces $\Var_\rho(H_{jk})$ to a finite-dimensional covariance--quadratic form. For global objectives, under an exponential concentration condition on the cost at initialization, $\Var_\rho(H_{jk})$ decays exponentially with the number of qubits $n$. For local averaged objectives in bounded-depth circuits, $\Var_\rho(H_{jk})$ admits polynomial bounds controlled by the growth of the backward lightcone on the interaction graph.
As a consequence, the number of measurement shots required to estimate $H_{jk}$ to fixed accuracy inherits the same exponential (global) or polynomial (local) scaling.
Extensive numerical experiments over system size, circuit depth, and interaction graphs validate the predicted variance scaling.
Overall, the paper quantifies when Hessian entries can be resolved at initialization under finite sampling, providing a mathematically grounded basis for second-order information in variational optimization.

\end{abstract}



\section{Introduction}
\label{sec:introduction}

Variational quantum algorithms (VQAs) have emerged as a dominant paradigm for the noisy intermediate-scale quantum (NISQ) era, formulating problems as stochastic optimization over parameterized quantum circuits~\cite{an2023linear, cerezo2021variational, hu2024quantum, jin2024quantum, kandala2017hardware, Peruzzo2014, zhu2025quantum}. However, their scalability is fundamentally constrained by the geometry of the optimization landscape. Beyond the general NP-hardness of training~\cite{PhysRevLett.127.120502}, a central obstruction is the \emph{barren plateau} phenomenon~\cite{larocca2025barren, mcclean2018barren}: for broad classes of ansatz families, gradient variances decay exponentially with system size, rendering first-order optimization ineffective at random initialization~\cite{bharti2022noisy, Grant2019initialization}.

The mechanisms driving this first-order information loss are now well understood. Theoretical analyses link vanishing gradients to factors such as excessive entanglement~\cite{marrero2021entanglement}, the globality of the cost operator~\cite{cerezo2021cost}, and hardware noise~\cite{wang2021noise}. Geometrically, these landscapes are characterized by narrow gorges~\cite{arrasmith2022equivalence} or a proliferation of traps~\cite{PhysRevA.111.012441}. In response, a diverse array of diagnostic~\cite{larocca2022diagnosing} and mitigation strategies has been developed, encompassing architectural designs~\cite{PhysRevX.11.041011}, specialized initialization techniques~\cite{Grant2019initialization, meng2026trainability}, layer-wise training~\cite{Skolik2021}, and controls on ansatz expressibility~\cite{holmes2022expressibility}.

Despite this progress, rigorous quantitative results for Hessian structure and scaling at initialization remain limited. Second-order methods, such as Newton-type updates~\cite{meng2026trainability} and the quantum natural gradient~\cite{stokes2020quantum}, are standard tools in scientific computing, but in the quantum setting their usefulness is constrained by how accurately Hessian information can be estimated as the system size grows. While prior work indicates that higher-order derivatives can also vanish~\cite{cerezo2021higher} and empirical studies have explored specific applications~\cite{sen2022hessian}, existing theoretical analyses largely focus on global operator norms or specific backpropagation schemes~\cite{Bowles2025BackpropScalingPQC}. These approaches do not quantify the \emph{entry-wise statistical resolvability} of the Hessian at random initialization, which is essential in practice where Hessian estimators are formed from finite-shot measurements. Consequently, a fundamental question arises: when first-order gradients vanish, does second-order information remain statistically accessible, and what precise scaling laws govern its variance?

In this work, we bridge this gap by quantifying the initialization scaling of Hessian-entry variance at random initialization.
Focusing on hardware-efficient architectures with bounded-depth locality, our analysis rests on two ingredients:
(i) exact parameter-shift rules~\cite{bergholm2018pennylane, crooks2019gradients, wierichs2022general} for derivative evaluation, and
(ii) the role of cost-function locality in mitigating first-order barren plateaus~\cite{cerezo2021cost, uvarov2021locality}.
Using exact second-order parameter-shift identities, we express each Hessian entry as a constant-size linear combination of shifted objective evaluations.
This representation rewrites the entrywise variance $\Var_\rho(H_{jk})$ as a finite covariance--quadratic form and makes its scaling explicitly controllable through backward-lightcone geometry.

\begin{figure}[htbp]
\centering
\includegraphics[width=0.8\linewidth]{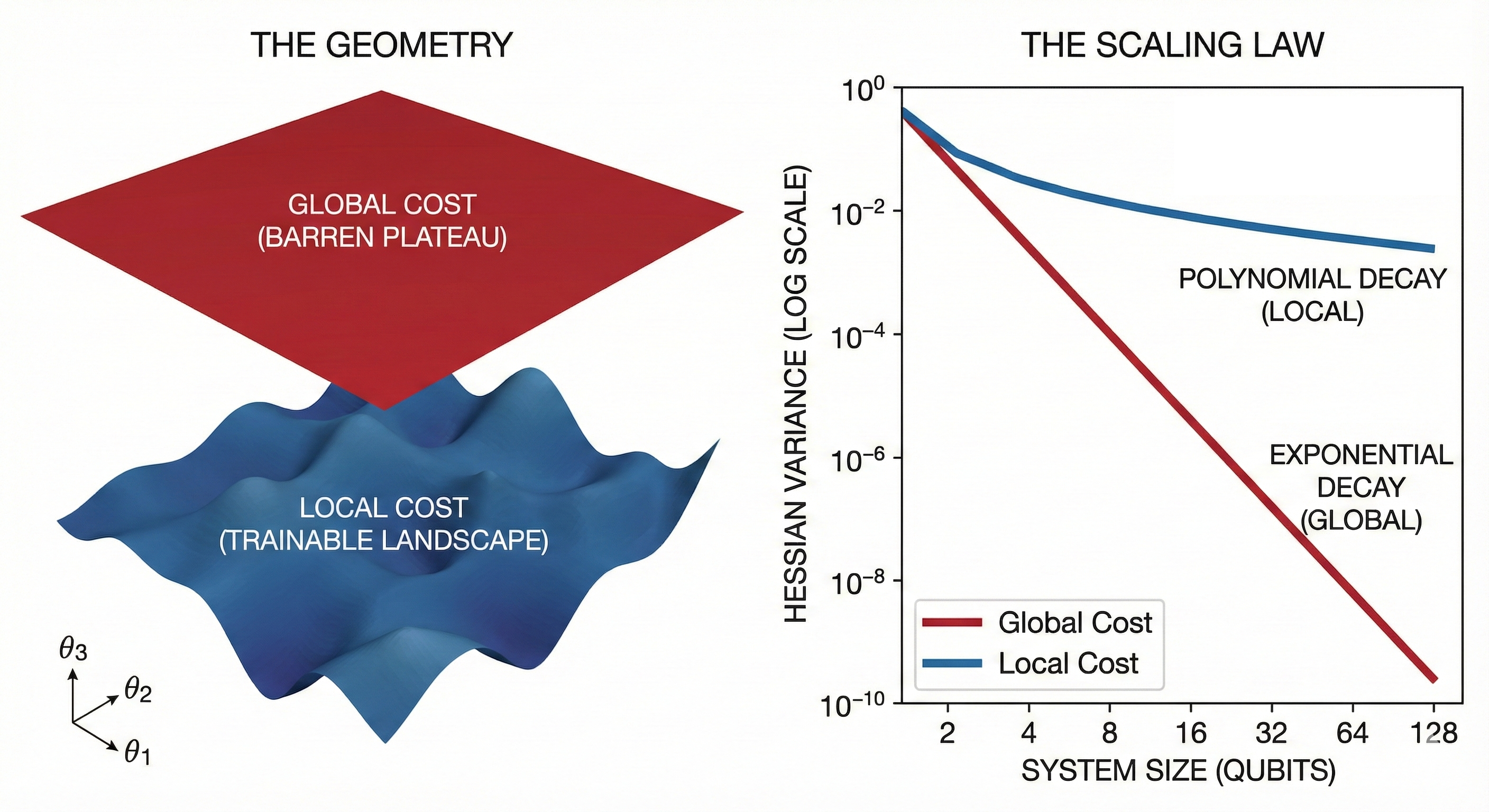}
\caption{\textbf{Higher-order barren plateaus: global vs.\ local objectives.}
Schematic landscapes and the expected scaling of Hessian-entry variance with system size: global costs exhibit exponential decay, whereas term-wise local costs exhibit polynomial decay.}
\label{fig:rel_improve}
\end{figure}

Our analysis yields two distinct scaling regimes at random initialization. Fig.~\ref{fig:rel_improve} schematically illustrates these regimes, which we refer to as \emph{higher-order barren plateaus}. For global objectives, we show that Hessian-entry variances inherit the exponential concentration of the objective, implying an exponentially growing finite-shot cost to resolve individual Hessian entries as the system size $n$ increases. In contrast, for term-wise $k$-local objectives (sums of terms acting on at most $k$ qubits) implemented by bounded-depth circuits, we obtain polynomial variance bounds controlled by the growth of the backward lightcone on the interaction graph. When the lightcone remains sub-extensive, these bounds yield a polynomial resolution cost for Hessian-entry estimation at initialization. We support the theory with numerical experiments that verify the predicted scaling laws, probe the Hessian eigenspectra at initialization, and quantify the corresponding sampling costs implied by our resolvability criterion.

The remainder of the paper is structured as follows. In \S~\ref{sec:background}, we establish the theoretical framework, including architectural assumptions and the relevant parameter-shift representations. \S~\ref{sec:analysis} presents our main theoretical analysis, deriving the explicit variance bounds and scaling laws that distinguish global from local objectives. These analytical predictions are demonstrated in \S~\ref{sec:experiments} through extensive numerical experiments. Finally, we discuss practical implications and future directions in \S~\ref{sec:conclusion}. Detailed proofs and supplementary technical results are provided in the appendices.

\section{Theoretical Background}
\label{sec:background}

We introduce the basic ingredients needed for the analysis of higher-order derivatives in VQAs. Specifically, the circuit model and objective functions, formalize locality notions that control operator spreading, and derive the parameter-shift identities for gradients and Hessian entries used throughout the paper. The initialization assumptions and the definition of higher-order barren plateaus are also stated here; these will be used to obtain the variance bounds in \S~\ref{sec:analysis}.

\subsection{Parameterized quantum circuits}
\label{sec:sub:pqc}

A parameterized quantum circuit is a family of unitaries $\{U(\bth)\}_{\bth\in\mathbb{R}^M}$ acting on the $n$-qubit Hilbert space $\mathcal{H}_n=(\mathbb{C}^2)^{\otimes n}$, where $M$ is the number of parameters. For each $\bth$, the circuit prepares
\begin{equation}
|\psi(\bth)\rangle = U(\bth)\,|0\rangle^{\otimes n}\in\mathcal{H}_n.
\end{equation}

We focus on hardware-efficient architectures of depth $L$ obtained by alternating fixed entangling layers and parameterized rotations. Specifically, each layer $\ell$ consists of a parameter-independent entangling unitary $V_\ell$ acting on $(\mathbb{C}^2)^{\otimes n}$, followed by $m_\ell$ parametrized gates, so that
\begin{equation}
\label{eq:pqc}
U(\bth)
=
\prod_{\ell=1}^{L}
\left(
V_{\ell}
\prod_{j=1}^{m_{\ell}}
\exp(-i\theta_{\ell,j}G_{\ell,j})
\right),
\end{equation}
with products ordered in increasing $\ell$ and $M=\sum_{\ell=1}^{L} m_{\ell}$. We assume that each parametrized gate in~\eqref{eq:pqc} carries its own independent parameter (i.e., no parameter sharing), so that the parameter vector can be written as ${\bth}=\{\theta_{\ell,j}\}_{\ell=1,\dots,L;\, j=1,\dots,m_\ell}$.

To enable exact constant-cost parameter-shift rules for gradients and Hessians, we assume each parametrized gate has a two-point-spectrum generator (e.g., Pauli rotations). We note that more general spectra admit generalized shift rules with a larger constant overhead.

\begin{assumption}[Two-point spectrum generators]
\label{ass:gen}
In \eqref{eq:pqc}, each generator $G_{\ell,j}$ is Hermitian and supported on at most $r$
qubits, where $r$ is independent of $n$. Moreover, $G_{\ell,j}$ has a two-point spectrum
$\mathrm{spec}(G_{\ell,j})=\{\pm \gamma\}$ for some $\gamma>0$. Without loss of generality we
take $\gamma=\tfrac12$, i.e., $G_{\ell,j}^2=\frac14 I$.
This includes Pauli rotations $G_{\ell,j}=P_{\ell,j}/2$ with $P_{\ell,j}$ a Pauli string
supported on at most $r$ qubits.
\end{assumption}

\subsection{Cost functions and locality}
\label{sec:sub:local}

The objective is to minimize the expectation value of a Hermitian observable $O$ under the variational state $|\psi(\bth)\rangle$:
\begin{equation}
\label{eq:cost}
C(\bth)=\langle \psi(\bth)|\,O\,|\psi(\bth)\rangle
       =\langle 0|\,U^{\dagger}(\bth)\,O\,U(\bth)\,|0\rangle.
\end{equation}
Locality of $O$ controls the spatial scale on which the conjugated observable $U^\dagger(\bth)OU(\bth)$ can spread, and thus the dependence structure underlying the variance bounds in \S~\ref{sec:analysis}.

\begin{definition}[Local versus global objectives]
\label{def:local_global}
An observable $O$ is \emph{$k$-local} if it acts non-trivially on at most $k$ qubits, i.e.,
$|\supp(O)|\le k$ for a constant $k$ independent of $n$. 
More generally, an objective $C(\bth)$ is \emph{$k$-local term-wise} if it can be written as
$C(\bth)=\frac{1}{n}\sum_{i=1}^n \langle\psi(\bth)|\,O_i\,|\psi(\bth)\rangle$ with each $O_i$
being $k$-local. 
We call an observable $O$ (and the objective it induces via \eqref{eq:cost}) \emph{global} if it
has extensive support, i.e., $|\supp(O)|$ grows linearly
with $n$.
\end{definition}

A canonical global example is
\begin{equation}
\label{eq:globalcost}
O_{\mathrm{glo}} = Z^{\otimes n},
\qquad
C_{\mathrm{glo}}(\bth)=\langle\psi(\bth)|Z^{\otimes n}|\psi(\bth)\rangle,
\end{equation}
which acts non-trivially on all qubits. In contrast, local objectives are typically constructed from bounded-support terms, e.g., $Z_i$ or $Z_i Z_{i+1}$. 
For normalization one often considers averaged local costs, for instance
\begin{equation}
\label{eq:localcost}
C_{\mathrm{loc}}(\bth)=\frac{1}{n}\sum_{i=1}^{n}\langle\psi(\bth)|Z_i|\psi(\bth)\rangle,
\end{equation}
which remains $1$-local term-wise while producing an $O(1)$-scaled objective.

\subsection{Parameter-shift identities}
\label{sec:sub:param_shift}

Under Assumption~\ref{ass:gen}, derivatives of $C(\bth)$ admit exact two-shift parameter-shift formulas~\cite{crooks2019gradients}. In particular,
\begin{equation}
\label{eq:psr1}
\frac{\partial C}{\partial \theta_j}
=
\frac{1}{2}\Big[
C(\bth + \tfrac{\pi}{2} \ej)
- C(\bth - \tfrac{\pi}{2} \ej)
\Big],
\end{equation}
where $\ej$ denotes the $j$-th canonical basis vector in $\mathbb{R}^M$. Applying \eqref{eq:psr1} twice yields the Hessian entries. For $j\neq k$,
\begin{align}
\label{eq:psr2-off}
\Hjk = \frac{1}{4}\Big[
C(\bth + \tfrac{\pi}{2}\ej &+ \tfrac{\pi}{2}\ek)
- C(\bth + \tfrac{\pi}{2}\ej - \tfrac{\pi}{2}\ek)
\nonumber \\
&- C(\bth - \tfrac{\pi}{2}\ej + \tfrac{\pi}{2}\ek)
+ C(\bth - \tfrac{\pi}{2}\ej - \tfrac{\pi}{2}\ek)
\Big],
\end{align}
and for $j=k$,
\begin{equation}
\label{eq:psr2-diag}
\Hjj = \frac{1}{4}\Big[
C(\bth + \pi \ej)
- 2 C(\bth)
+ C(\bth - \pi \ej)
\Big].
\end{equation}
For more general generator spectra, generalized parameter-shift rules still yield finite linear combinations of shifted costs~\cite{wierichs2022general}. We note that our results only require an $O(1)$ number of shifts per derivative.

\subsection{Trainability metrics at random initialization}
\label{sec:sub:assumptions}

We study the initialization-time variance of gradient and Hessian entries. Unless otherwise stated, all expectations and variances are taken with respect to ${\bth}\sim\rho$, where $\rho$ is the product measure of i.i.d.\ uniform random variables on $[0,2\pi)$.

\begin{assumption}[Depth regime]
\label{ass:depth}
The circuit depth $L$ is either fixed or grows sublinearly with $n$; in particular, $rL=o(n)$ as $n\to\infty$, where $r$ is the gate locality from Assumption~\ref{ass:gen} (each generator acts on at most $r$ qubits).
\end{assumption}

Under Assumptions~\ref{ass:gen} and~\ref{ass:depth}, and for a $k$-local observable $O$ in Definition~\ref{def:local_global}, conjugation by $U(\bth)$ expands the support of $O$ only within a backward lightcone of radius $O(rL)$. We will quantify the corresponding neighborhood growth via the graph growth function $V_G(\cdot)$ introduced in \S~\ref{subsec:main_result}. 

\begin{assumption}[Observable normalization]
\label{ass:obs}
The observable $O$ in \eqref{eq:cost} is Hermitian and satisfies $\|O\|\le 1$. For term-wise local objectives $C(\bth)=\frac{1}{n}\sum_{i=1}^n \langle\psi(\bth)|O_i|\\\psi(\bth)\rangle$, we also assume $\|O_i\|\le 1$.
\end{assumption}

This normalization ensures that asymptotic variance rates reflect circuit geometry rather than trivial rescaling. We therefore quantify trainability at initialization through variance-based resolvability metrics as follows.

\subsubsection{Variance of Hessian entries}
For any indices $(j,k)$, we quantify the initialization-time variability of a Hessian entry by
\begin{equation}
\label{eq:varH}
\Var_\rho(H_{jk})
= \E_\rho[H_{jk}^2] - \big(\E_\rho[H_{jk}]\big)^2.
\end{equation}
We interpret $\Var_\rho(H_{jk})$ as a proxy for the \emph{statistical resolvability} of the entry $H_{jk}$ under finite-shot estimation: smaller variance indicates that substantially more samples are required to distinguish the entry from sampling noise at a prescribed accuracy.

\begin{definition}[Higher-order barren plateau]
\label{def:HOBP}
Let $C(\bth)$ be the cost \eqref{eq:cost} with Hessian $H(\bth)$. We say the circuit exhibits a \emph{higher-order barren plateau at initialization (entry-wise)} if there exist constants $c,\alpha>0$ independent of $n$ such that
\begin{equation}
\label{eq:exp_decay}
\Var_\rho(H_{jk}) \le c\, e^{-\alpha n}
\quad \text{for all } j,k \text{ and all sufficiently large } n.
\end{equation}
\end{definition}

Definition~\ref{def:HOBP} is entry-wise and does not by itself imply spectral conditioning (e.g.,
bounds on extremal eigenvalues or condition numbers); small typical entries can still coexist with a
few directions of non-negligible curvature. Nevertheless, entry-wise second-moment control is the
natural notion for finite-sample shift-based estimation and yields, for instance, $\E_\rho[\|H\|_2^2]\le \E_\rho[\|H\|_F^2]=\sum_{j,k}\E_\rho[H_{jk}^2]$ (cf.~Lemma~\ref{lem:entrywise-to-norms}).

The following lemma reduces Hessian-entry variance to a covariance--quadratic form over finitely many shifted objective evaluations, which is the starting point for our global/ local scaling analysis. The proof is given in Appendix~\ref{app:proof_var_rep}.

\begin{lemma}[Parameter-shift covariance representation]
\label{lem:var_rep}
Under Assumption~\ref{ass:gen}, the diagonal Hessian entry admits the second-order parameter-shift representation~\eqref{eq:psr2-diag}.
Define $C_{+}:=C(\bth+\pi\ej)$, $C_{0}:=C(\bth)$, and $C_{-}:=C(\bth-\pi\ej)$, and let
$\Sigma\in\mathbb{R}^{3\times 3}$ be the covariance matrix with entries $\Sigma_{\alpha\beta}:=\Cov_\rho(C_\alpha,C_\beta)$ for $\alpha,\beta\in\{+,0,-\}$.
Then
\begin{equation}
\label{eq:var_Hjj_quad}
\Var_\rho(H_{jj}) = w^\top \Sigma w,
\qquad
w:=\tfrac14(1,-2,1)^\top,
\end{equation}
and equivalently $\Var_\rho(H_{jj})$ can be written as an explicit linear combination of variances and covariances among $\{C_+,C_0,C_-\}$. 
An analogous covariance--quadratic representation holds for off-diagonal entries $H_{jk}$ using \eqref{eq:psr2-off}.
More generally, whenever a derivative quantity admits a finite-shift representation
$X(\bth)=\sum_{\ell=1}^m w_\ell\, C(\bth\oplus s_\ell)$ with $m<\infty$, one has
$\Var_\rho(X)=\sum_{\ell,\ell'=1}^m w_\ell w_{\ell'}\Cov_\rho\!\big(C(\bth\oplus s_\ell),C(\bth\oplus s_{\ell'})\big)$.
\end{lemma}

Lemma~\ref{lem:var_rep} extends to any differentiation scheme that represents $H_{jk}$ as a finite linear combination of shifted costs, provided the number of shifted evaluations is $O(1)$, with constants determined by the shift coefficients.

\section{Asymptotic Variance Analysis}
\label{sec:analysis}

This section derives asymptotic bounds for the entry-wise Hessian variance at random initialization. 
We first express $\Var_\rho(H_{jk})$ as a finite covariance--quadratic form over a constant number of shifted cost evaluations.
We then bound these covariances using locality and backward-lightcone structure for local objectives, and using a cost-concentration condition for global objectives.

\subsection{Variance representation}
\label{subsec:var_rep}

Let $C(\bth)$ be the cost~\eqref{eq:cost}. The second-order parameter-shift rules~\eqref{eq:psr2-off}--\eqref{eq:psr2-diag} express each Hessian entry as a finite linear combination of shifted costs, and hence yield the covariance--quadratic form 
\begin{equation}
\label{eq:var_quadratic_form_intro}
\Var_\rho(H_{jk})
= \sum_{s,s'\in\mathcal S_{jk}} w_s w_{s'}\Cov_\rho\big(C(\bth \oplus s),C(\bth \oplus s')\big),
\end{equation}
where $\oplus$ denotes component-wise addition modulo $2\pi$.
For the two-shift rules \eqref{eq:psr2-off}--\eqref{eq:psr2-diag}, one has $|\mathcal S_{jk}|\le 4$ and $w_s=O(1)$ independent of $n$. Derived from Lemma~\ref{lem:var_rep}, the explicit diagonal and off-diagonal shift rules are listed in Appendix~\ref{app:cov_quad_ps}; \eqref{eq:var_quadratic_form_intro} is the starting point for the variance bounds in \S~\ref{subsec:main_result}.

While \eqref{eq:var_quadratic_form_intro} characterizes the initialization-induced variability of $H_{jk}$, practical estimation relies on finite-shot evaluations of the shifted costs. It is therefore natural to compare the shot-noise contribution to the estimator variance with the intrinsic scale $\Var_\rho(H_{jk})$ and to quantify the number of shots needed to separate the two. This motivates the following resolution criterion.

\begin{definition}[Resolution criterion and shot requirement]
\label{def:resolution}
Let $\widehat H_{jk}$ be a finite-shot estimator of $H_{jk}$ constructed from shifted cost evaluations. 
For $\eta\in(0,1]$, we say that $H_{jk}$ is \emph{resolved at initialization} if
\[
\E_{\bth\sim\rho}\left[\Var_{\mathrm{sh}}(\widehat H_{jk}\mid\bth)\right]\le \eta\Var_\rho(H_{jk}).
\]
The corresponding \emph{resolution shot requirement} $N_{\mathrm{res}}$ is the smallest number of shots per shifted cost evaluation for which the above condition holds.
\end{definition}

Algorithm~\ref{alg:psr_hessian} states the finite-shot parameter-shift estimator used throughout and in the experiments of \S~\ref{sec:experiments}.

\begin{algorithm}[t]
\caption{Finite-shot estimation of a Hessian entry via parameter-shift}
\label{alg:psr_hessian}
\begin{center}
\begin{minipage}{0.9\linewidth}
\begin{algorithmic}[1]
\Require Ansatz $U(\bth)$, objective $C(\bth)$, indices $(j,k)$, measurement shots $N$
\Ensure Unbiased estimate $\widehat H_{jk}$ of the Hessian entry $H_{jk}$

\State $\widehat H_{jk} \gets 0$
\ForAll{$s \in \mathcal S_{jk}$} \Comment{Loop over shift set $\mathcal S_{jk}$ with coeff. $w_s$ (cf. \S~\ref{sec:sub:param_shift})}
    \State Prepare $|\psi(\bth \oplus s)\rangle$ and estimate $C(\bth\oplus s)=\langle O\rangle$ using $N$ shots
    \State $\widehat C(\bth \oplus s) \gets$ sample mean of the measurement outcomes
    \State $\widehat H_{jk} \gets \widehat H_{jk} + w_s\widehat C(\bth \oplus s)$ 
\EndFor
\State \Return $\widehat H_{jk}$
\end{algorithmic}
\end{minipage}
\end{center}
\end{algorithm}

\subsection{Main results}
\label{subsec:main_result}

We now derive explicit scaling bounds from~\eqref{eq:var_quadratic_form_intro} by bounding covariances between shifted costs. We treat global objectives first, where second-moment concentration transfers through the constant-size finite-shift rules, and then turn to term-wise $k$-local averaged objectives.

\begin{assumption}[Second-moment concentration for global objectives]
\label{ass:GlobalDesign}
Let $O$ be Hermitian with $\|O\|\le 1$ and define the cost $C(\bth)$ by \eqref{eq:cost}.
Assume there exist constants $c_{\rm g}(r,L)>0$ and $\alpha(r,L)>0$, independent of $n$, such that
\begin{equation}
\label{eq:global_concentration}
\Var_\rho\big[C(\bth)\big]\le c_{\rm g}(r,L)e^{-\alpha(r,L)n}.
\end{equation}
\end{assumption}

\begin{remark}
Assumption~\ref{ass:GlobalDesign} is a second-moment concentration condition on the objective value at random initialization. Its role is to serve as an input for the global regime; Lemma~\ref{lem:transference} shows that such concentration is uniform over the finite shift sets used in parameter-shift rules and is preserved under any constant-size finite-shift linear combination. 

The assumption holds, for example, when the initialized state ensemble forms an (approximate) projective $2$-design, in which case $\Var_\rho[C(\bth)]$ decays as $\mathcal{O}(2^{-n})$ for $\|O\|\le 1$ (see Appendix~\ref{app:design_route} and Proposition~\ref{prop:2design_implies_global}). In architectures where design guarantees are unavailable at the considered depth, \eqref{eq:global_concentration} can be assessed directly by Monte Carlo estimation of $\Var_\rho[C(\bth)]$; by shift-invariance of $\rho$, the same estimate applies to all shifted evaluations appearing in the parameter-shift rules.
\end{remark}

The next lemma formalizes this transference: second-moment concentration of $C(\bth)$ is uniform over shifts and stable under finite-shift linear combinations.

\begin{lemma}[Uniform shift concentration and finite-shift transference]
\label{lem:transference}
Let $\rho$ be shift-invariant on the torus, i.e., $\bth\sim\rho \Rightarrow \bth\oplus s\sim\rho$ for all shifts $s\in\mathbb{R}^M$.
Fix a finite shift set $\mathcal S\subset\mathbb{R}^M$ and coefficients $\{w_s\}_{s\in\mathcal S}\subset\mathbb{R}$, and define
$X_s := C(\bth\oplus s)$. If $\Var_{\rho}[C(\bth)]\le v_n$, then the same bound holds uniformly over shifts,
\[
\Var_{\rho}[X_s]=\Var_{\rho}[C(\bth\oplus s)] = \Var_{\rho}[C(\bth)] \le v_n
\qquad \text{for all } s\in \mathcal S,
\]
and moreover,
\[
\Var_{\rho}\Big(\sum_{s\in \mathcal S} w_s\,C(\bth\oplus s)\Big)
\le
\Big(\sum_{s\in \mathcal S}|w_s|\Big)^2 v_n .
\]
\end{lemma}

\begin{proof}
Shift-invariance implies $\bth$ and $\bth\oplus s$ have the same distribution under $\rho$, hence
$\Var_{\rho}[C(\bth\oplus s)] = \Var_{\rho}[C(\bth)]$ for all $s\in \mathcal S$. For the linear combination, by bilinearity,
\[
\Var\Big(\sum_{s\in \mathcal S} w_s X_s\Big)=\sum_{s,s'\in \mathcal S} w_s w_{s'}\Cov(X_s,X_{s'}).
\]
By Cauchy--Schwarz,
$|\Cov(X_s,X_{s'})|\le \sqrt{\Var(X_s)\Var(X_{s'})}\le v_n$,
where the last inequality uses the uniform variance bound from the first part.
Therefore,
\[
\Var\Big(\sum_{s\in \mathcal S} w_s X_s\Big)
\le v_n\sum_{s,s'\in \mathcal S}|w_s||w_{s'}|
= \Big(\sum_{s\in \mathcal S}|w_s|\Big)^2 v_n.
\]
\end{proof}

To make the locality bound explicit, we quantify how many local terms can have overlapping
backward lightcones, which controls the maximum degree of the associated dependency graph.
Let $G$ be the interaction graph of the architecture and denote by $B_G(v,m)$ the ball of radius $m$
centered at $v$. Define the growth function
\begin{equation}
\label{eq:graph_growth}
V_G(m):=\max_{v\in V(G)} \big|B_G(v,m)\big|.
\end{equation}
For bounded-degree graphs, $V_G(m)$ is finite for each fixed $m$; for $D$-dimensional lattices one has
$V_G(m)=\mathcal{O}(m^D)$. We are then ready to state our main theorem.

\begin{theorem}[Asymptotic scaling of Hessian-entry variance]
\label{thm:scaling-explicit}
Under Assumptions~\ref{ass:gen}, \ref{ass:depth}, and \ref{ass:obs}, the Hessian-entry variances satisfy, for all $j,k\in\{1,\ldots,M\}$:
\begin{enumerate}[(i)]
\item Global objectives. If $|\supp(O)|$ grows linearly with $n$ and Assumption~\ref{ass:GlobalDesign} holds, then there exist constants $\tilde c(r,L)>0$ and $\tilde\alpha(r,L)>0$ such that
\begin{equation}
\label{eq:global_var_bound}
\Var_\rho(H_{jk}) \le \tilde c(r,L)\exp\big(-\tilde\alpha(r,L) n\big).
\end{equation}

\item Local objectives. Suppose the cost is of the form
\[
C_{\rm loc}(\bth)=\frac{1}{n}\sum_{v=1}^{n}\langle\psi(\bth)|O_v|\psi(\bth)\rangle,
\qquad \|O_v\|\le 1,
\]
where each $O_v$ is supported on at most $k$ qubits with $k$ independent of $n$. Then there exists $c_{\rm loc}(k,r,L)>0$ such that
\begin{equation}
\label{eq:local_bound}
\Var_\rho(H_{jk}) \le c_{\rm loc}(k,r,L)\frac{V_G(k+2rL)}{n}.
\end{equation}
In particular, if $V_G(m)=\mathcal{O}(m^D)$ for some fixed $D$, then
\begin{equation}
\label{eq:local_bound_poly_growth}
\Var_\rho(H_{jk}) \le c_{\rm loc}(k,r,L)\frac{(k+2rL)^D}{n}.
\end{equation}
\end{enumerate}
\end{theorem}

\begin{proof}[Proof sketch]
By the parameter-shift representation, each Hessian entry $H_{jk}$ is a fixed $O(1)$-term linear combination of shifted costs $C(\bth\oplus s)$, hence $\Var_\rho(H_{jk})$ reduces to a finite covariance--quadratic form over a constant-size shift set.

For global objectives, Lemma~\ref{lem:transference} transfers the assumed second-moment concentration of $C(\bth)$ uniformly over the finite shifts; since the parameter-shift rule uses only $O(1)$ evaluations with $O(1)$ coefficients, the exponential rate in $n$ is preserved up to a constant factor, yielding~\eqref{eq:global_var_bound}.

For local objectives, bounded lightcone growth implies that sufficiently separated local contributions depend on disjoint subsets of parameters; under the product initialization measure $\rho$, their covariances vanish. A dependency-graph argument then gives $\Var_\rho(C(\bth\oplus s))=\mathcal{O}(V_G(k+2rL)/n)$ uniformly in $s$. Plugging this bound into the covariance representation yields~\eqref{eq:local_bound}. Details are deferred to Appendix~\ref{app:depgraph}.
\end{proof}

\begin{remark}
\label{rem:depth-scaling-threshold}
The local bound in Theorem~\ref{thm:scaling-explicit} depends on depth only through the lightcone growth factor
$V_G(k+2rL)$. Hence a sufficient condition for decay with system size is that the relevant backward lightcone remains
\emph{subextensive}, $V_G(k+2rL)=o(n)$. For bounded-degree $D$-dimensional lattices, $V_G(s)=O(s^D)$ up to the graph diameter and then reaches system size.
Therefore the bound remains vanishing provided $k+2rL=o(n^{1/D})$ (equivalently, $L=o(n^{1/D})$ up to constants),
whereas once $k+2rL$ reaches the linear system size the lightcone becomes extensive and the locality-based suppression
may no longer decay with $n$. On small-diameter graphs (e.g., complete graphs), $V_G$ can saturate rapidly, so this
locality protection can disappear at very shallow depth. For example, in 1D chains this permits $L=o(n)$, while in 2D grids it requires $L=o(\sqrt{n})$.
\end{remark}

\begin{remark}
The estimate \eqref{eq:local_bound} is an upper bound obtained via a generic dependency-graph argument and is not expected to be exponent-tight in typical instances.
The main message of Theorem~\ref{thm:scaling-explicit} is that the measurement cost of resolving Hessian entries is controlled by lightcone growth: exponential decay can occur for global objectives, while bounded-depth locality yields polynomial bounds.
\end{remark}

\subsection{From entry-wise variance to matrix norms and sample complexity}
\label{sec:resolvability-vs-conditioning}

While the entry-wise bounds above quantify the typical fluctuation scale, the practical utility of curvature information depends on the number of shots required to separate these fluctuations from measurement noise, and on how they aggregate into matrix-level quantities.
Throughout this subsection, $\E_\rho[\cdot]$ and $\Var_\rho(\cdot)$ denote expectation and variance with respect to initialization $\bth\sim\rho$. When finite-shot estimation is involved, we write $\E_{\mathrm{sh}}[\cdot\mid\bth]$ and $\Var_{\mathrm{sh}}(\cdot\mid\bth)$ for expectation and variance over measurement noise (shots), conditional on a fixed parameter value $\bth$.

We first record a standard implication of entry-wise second-moment control for matrix norms; a proof is included in Appendix~\ref{app:norm_bounds} for completeness.

\begin{lemma}[Entry-wise second moments imply Frobenius and spectral norm bounds]
\label{lem:entrywise-to-norms}
Let $H\in\mathbb{R}^{M\times M}$ be a random matrix with finite second moments. Then
\[
\E\|H\|_F^2=\sum_{j,k=1}^M \E[H_{jk}^2],
\qquad
\E\|H\|_2^2 \le \E\|H\|_F^2,
\qquad
\E\|H\|_2 \le \sqrt{\E\|H\|_F^2}.
\]
In particular, if $\Var(H_{jk})\le v_n$ and $\big|\E [H_{jk}]\big|\le \mu_n$ for all $j,k$, then
\[
\E\|H\|_F^2 \le M^2(v_n+\mu_n^2),
\qquad
\E\|H\|_2 \le M\sqrt{v_n+\mu_n^2}.
\]
\end{lemma}

Lemma~\ref{lem:entrywise-to-norms} converts entry-wise second-moment bounds into a global \emph{norm scale} for the Hessian at initialization. These estimates quantify typical magnitude (e.g., for Hessian--vector products) but do not, by themselves, control conditioning or extremal eigenvalues without additional spectral information. Applying this tool to the scalings established in Theorem~\ref{thm:scaling-explicit} yields the following two corollaries.

\begin{corollary}
\label{cor:matrix-scale-from-entrywise}
Under the conditions of Theorem~\ref{thm:scaling-explicit}. Suppose further that there exists a constant $c>0$,
independent of $n$, such that $|\E_\rho [H_{jk}]|^2 \le c\Var_\rho(H_{jk})$ for all $j,k$. Let $v_n$ denote the
entry-wise variance bound from Theorem~\ref{thm:scaling-explicit}. Then
\[
\E_\rho\|H\|_2^2 \le \E_\rho\|H\|_F^2 \lesssim M^2 v_n,
\qquad
\E_\rho\|H\|_2 \lesssim M\sqrt{v_n},
\]
with constants depending only on $c$.
\end{corollary}

While Corollary~\ref{cor:matrix-scale-from-entrywise} bounds the intrinsic curvature scale, the key practical question is the resolution cost in Definition~\ref{def:resolution} under finite-shot sampling.

\begin{corollary}[Absolute-accuracy shot cost at fixed parameters]
\label{cor:sample_complexity}
Let $\widehat H_{jk}(\bth)$ be the estimator computed via Algorithm~\ref{alg:psr_hessian} using $N$ independent shots per circuit evaluation. 
Since the parameter-shift rule is a fixed linear combination, the shot-noise contribution to the estimator variance satisfies
\begin{equation}
\Var_{\mathrm{sh}}\!\left(\widehat H_{jk}(\bth)\mid\bth\right) \le \frac{c_{\mathcal S}\sigma^2}{N},
\end{equation}
where $\sigma^2 \le 1$ when $\|O\|\le 1$ and $c_{\mathcal S} := \sum_{s \in \mathcal S_{jk}} w_s^2$ depends only on the chosen shift rule.
Consequently, achieving $\mathrm{RMSE}\le \varepsilon$ for estimating $H_{jk}(\bth)$ at fixed $\bth$ requires $N \gtrsim c_{\mathcal S}\sigma^2/\varepsilon^2$.
This bound concerns \emph{absolute} accuracy at fixed $\bth$; resolvability relative to the initialization scale is addressed in Definition~\ref{def:resolution} and Corollary~\ref{cor:resolution_cost}.
\end{corollary}

\begin{corollary}[Resolution shot requirement]
\label{cor:resolution_cost}
Under the assumptions of Corollary~\ref{cor:sample_complexity}, the resolution criterion in Definition~\ref{def:resolution} is ensured whenever
\[
N \gtrsim \frac{c_{\mathcal S}\sigma^2}{\eta\Var_\rho(H_{jk})}.
\]
In particular, combining this with the variance bounds in Theorem~\ref{thm:scaling-explicit} yields exponential scaling of the resolution cost for global objectives and polynomial scaling for term-wise local objectives (under sub-extensive backward-lightcones).
\end{corollary}

Combining Theorem~\ref{thm:scaling-explicit} with Corollary~\ref{cor:resolution_cost} shows that the shot requirement for resolving Hessian entries is governed by the initialization variance $\Var_\rho(H_{jk})$. For global objectives, exponential decay of $\Var_\rho(H_{jk})$ forces an exponentially growing resolution cost in $n$. For term-wise local objectives in bounded-depth circuits, $\Var_\rho(H_{jk})$ is bounded by a polynomial controlled by backward-lightcone growth, leading to polynomial resolution cost when the lightcone remains sub-extensive.

\section{Numerical Experiments}
\label{sec:experiments}

This section validates the variance-scaling predictions of \S~\ref{sec:analysis} and studies the empirical implications of these scaling laws under finite-shot sampling. After describing the simulation setting and variance estimators, we present: (i) system-size and depth scaling of representative Hessian-entry variances, including a VQE benchmark for a $k$-local Hamiltonian; (ii) spectral diagnostics of $H(\bth)$ at initialization, including histogram estimates and scalar summary metrics; and (iii) finite-shot optimization trajectories with matched oracle calls to illustrate the practical effect of the resolution criterion (Definition~\ref{def:resolution}) near initialization.

\subsection{Experimental setup}
\label{subsec:setup}

All simulations were performed in \texttt{Python} using \texttt{PennyLane}~\cite{bergholm2018pennylane} with an exact statevector backend, which ensures that the reported statistics isolate intrinsic landscape fluctuations. We instantiate \eqref{eq:pqc} as a layered hardware-efficient ansatz: in each layer $\ell$, $V_\ell$ is a fixed nearest-neighbor brickwork CNOT pattern and the parametrized gates are single-qubit $R_y$ rotations,
\[
\exp(-i\theta_{\ell,j}G_{\ell,j}) = R_y(\theta_{\ell,j}), 
\qquad 
G_{\ell,j}=\tfrac{1}{2}Y_{q(\ell,j)},
\]
where $q(\ell,j)$ maps the gate index $j\in\{1,\dots,m_\ell\}$ to a qubit wire. Here $m_\ell=n$ and the total number of parameters is $M=\sum_{\ell=1}^L m_\ell=nL$. Parameters are initialized i.i.d.\ as $\theta_{\ell,j}\sim \mathrm{U}([-\pi,\pi))$, which is equivalent to $\mathrm{U}([0,2\pi))$ up to a global shift.

We consider two objectives of different locality,
\begin{align}
C_{\mathrm{global}}(\bth) = \langle \psi(\bth)|Z^{\otimes n}| \psi(\bth) \rangle, \quad 
C_{\mathrm{local}}(\bth)  = \frac{1}{n} \sum_{j=1}^{n} \langle \psi(\bth) | Z_j | \psi(\bth) \rangle.
\end{align}
For each sampled initialization $\bth$ we evaluate gradients and Hessian entries via the parameter-shift identities \eqref{eq:psr1}--\eqref{eq:psr2-diag} on the statevector backend. Unless otherwise stated, the scaling experiments use $n \in \{2,\dots,16\}$, $L \in \{1,\dots,12\}$, and $N_s=200$ i.i.d.\ initializations per configuration $(n,L)$; variances are computed as empirical second central moments over this ensemble. Results are shown for $(j,k)=(1,2)$; other index choices exhibit the same scaling up to constant prefactors.

We estimate $\Var_\rho(\cdot)$ by the sample variance $\widehat{\Var}$ computed from $N_s$ i.i.d.\ random initializations $\bth\sim\rho$.
To extract scaling with system size $n$, we fit exponential laws by least squares on the log scale and polynomial laws by least squares in log--log form. When reported, error bars denote 95\% bootstrap confidence intervals for $\widehat{\Var}$ obtained by resampling the $N_s$ initializations.

\subsection{Scaling of Hessian variance}
\label{sec:sub:sub:scaling}

We first validate the system-size and depth scaling predicted by Theorem~\ref{thm:scaling-explicit}.

\subsubsection{Dependence on system size}
We examine how representative Hessian-entry variances at random initialization scale with the number of qubits $n$.

As a preliminary check for the global case (cf.~Assumption~\ref{ass:GlobalDesign}), we estimate the objective fluctuation $\Var_\rho[C_{\mathrm{global}}(\bth)]$
from the same initialization ensembles. Over $n\in\{2,\ldots,16\}$ the data are well described by an exponential fit, yielding $\Var_\rho[C_{\mathrm{global}}(\bth)] \\\approx \exp(-0.61\, n)$, which numerically verifies the assumption.

Fig.~\ref{fig:scaling} then reports empirical variances of a diagonal and an off-diagonal entry, $\Var_\rho(H_{jj})$ and
$\Var_\rho(H_{jk})$, as functions of $n$, for both the global and local objectives. For $C_{\mathrm{global}}$ (red markers), both diagonal and off-diagonal entries exhibit clear exponential decay with $n$.
Least-squares fits of $\log \Var_\rho(H_{jj})$ and $\log \Var_\rho(H_{jk})$ yield
\[
\Var_\rho(H_{jj}) \approx \exp(-0.57\, n), \qquad 
\Var_\rho(H_{jk}) \approx \exp(-0.64\, n),
\]
consistent with Theorem~\ref{thm:scaling-explicit} under Assumption~\ref{ass:GlobalDesign}. Via
Corollary~\ref{cor:sample_complexity}, this implies that the target resolution required to track typical entry-wise
curvature becomes exponentially small as $n$ grows.

For $C_{\mathrm{local}}$ (blue markers), the scaling is qualitatively different. Fitting $\log \Var_\rho(H_{jj})$ and
$\log \Var_\rho(H_{jk})$ against $\log n$ yields power-law behavior over the tested range,
\[
\Var_\rho(H_{jj}) \approx n^{-1.96}, \qquad 
\Var_\rho(H_{jk}) \approx n^{-2.05},
\]
indicating polynomial suppression consistent with the locality-controlled regime of Theorem~\ref{thm:scaling-explicit}.
We emphasize that \eqref{eq:local_bound} is an upper bound; the observed exponents should be interpreted as
instance-dependent rather than universal.

\begin{figure}[htbp]
    \centering
    \includegraphics[width=0.45\textwidth]{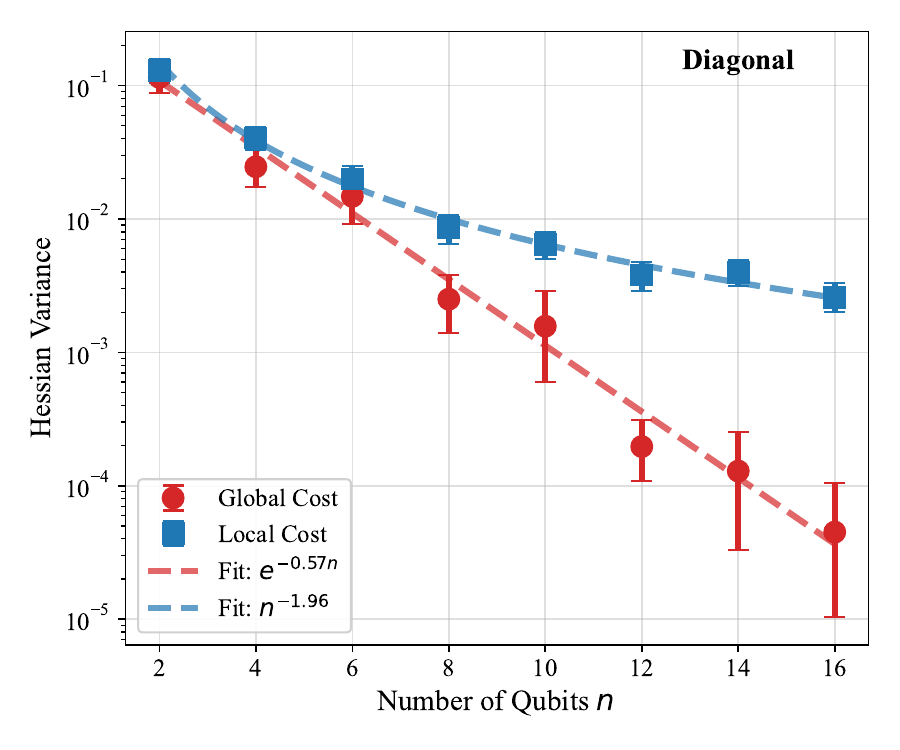}~~
    \includegraphics[width=0.45\textwidth]{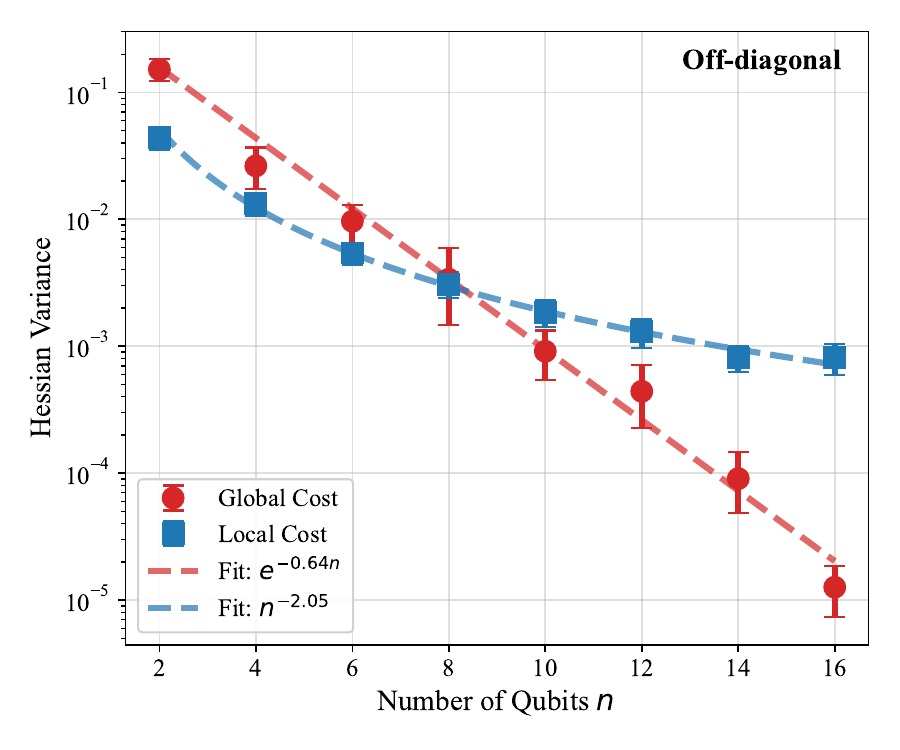}~~
    \caption{\textbf{Hessian-entry variance vs.\ system size.} $\Var[H_{jj}]$ (left) and $\Var[H_{jk}]$ (right) versus $n$ for $C_{\mathrm{global}}$ (red) and $C_{\mathrm{local}}$ (blue).}
    \label{fig:scaling}
\end{figure}

\subsubsection{VQE with Transverse Field Ising Model}

To validate our theoretical findings in a realistic variational quantum eigensolver (VQE) setting~\cite{kandala2017hardware}, we replace the toy objective (cf.~\S~\ref{subsec:setup}) with the Hamiltonian of the 1D Transverse Field Ising Model (TFIM)~\cite{pfeuty1970one} under periodic boundary conditions:
\begin{equation}
    H_{\text{TFIM}} = -J \sum_{i=1}^n Z_i Z_{i+1} - h \sum_{i=1}^n X_i,
\end{equation}
where we set $J=h=1.0$ (the critical point). The local cost function is defined as the term-wise averaged energy density $C_{\text{local}} = \langle H_{\text{TFIM}} \rangle / (2n)$. We compare this against a global parity observable $O = \bigotimes_{i=1}^n Z_i$. The variances of the Hessian entries are computed over $150$ random initializations for system sizes $n \in [4, 14]$.

\begin{figure}[t]
    \centering
    \includegraphics[width=0.45\textwidth]{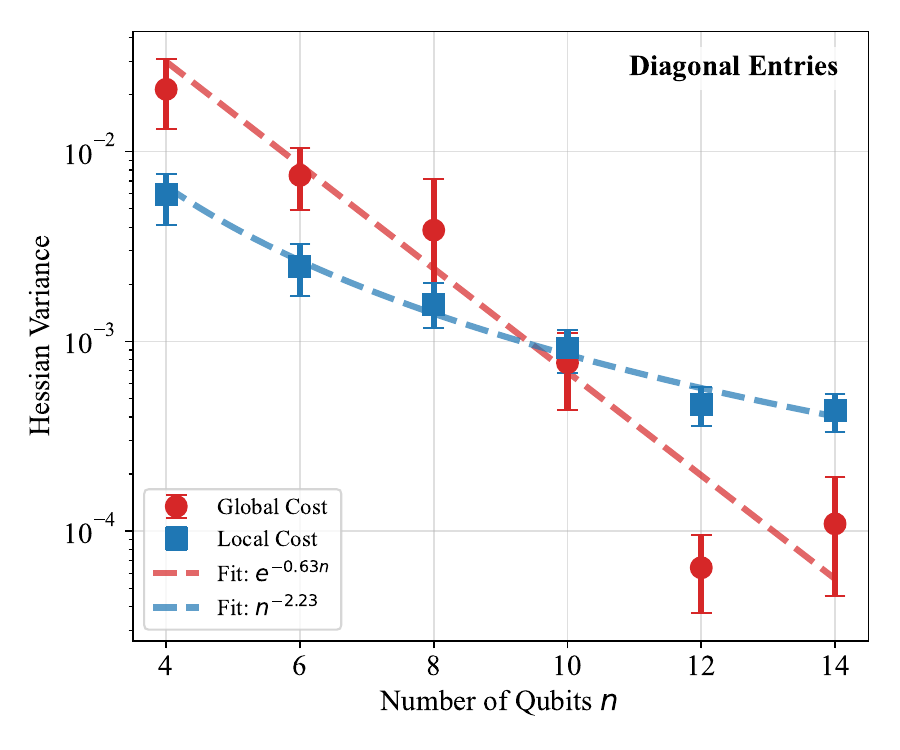}~~
    \includegraphics[width=0.45\textwidth]{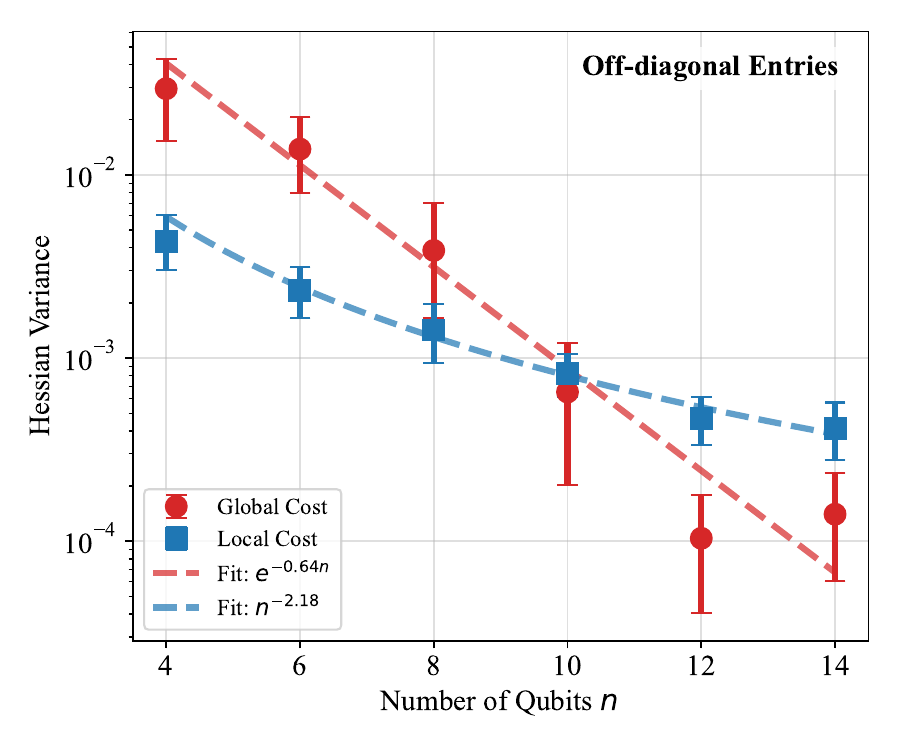}~~
    \caption{\textbf{Hessian-entry variance vs.\ system size in a VQE task (TFIM).} 
    Comparison of the variances of diagonal (left) and off-diagonal (right) Hessian entries between the VQE local cost (blue squares) and the global parity cost (red circles). }
    \label{fig:vqe_scaling}
\end{figure}

The results are presented in Fig.~\ref{fig:vqe_scaling}. The difference in scaling behavior is distinct: the global cost variance decays exponentially ($\mathrm{Var}\propto e^{-0.63n}$), rapidly hitting the precision floor, whereas the local VQE cost follows a clear power-law decay ($\mathrm{Var}\propto n^{-2.2}$), verifying that the polynomial scaling guarantee of Theorem~\ref{thm:scaling-explicit} holds for physical Hamiltonians with non-commuting terms. Consistent behavior is observed for both representative diagonal and off-diagonal Hessian entries.

\subsubsection{Dependence on depth}

We next examine how Hessian-entry fluctuations depend on circuit depth $L$, which provides a numerical view of how increasing operator spreading affects second-order resolvability. Fig.~\ref{fig:depth} reports empirical variances of representative diagonal and off-diagonal entries at fixed system size $n=16$ over the depth range $L\in\{1,\ldots,12\}$.

For $C_{\mathrm{global}}$ (red markers), both $\Var_\rho(H_{jj})$ and $\Var_\rho(H_{jk})$ are already on the order of $10^{-5}$ by depth $L=2$ and exhibit only weak additional dependence on $L$ over the explored range. This behavior is consistent with the global-support regime: backward lightcones overlap extensively even at shallow depth, so cancellations can suppress entry-wise curvature fluctuations early, leaving little variation to resolve at the scale of our reported statistics.

For $C_{\mathrm{local}}$ (blue markers), the variances are substantially larger at small depth (exceeding the global values by more than two orders of magnitude at $L=2$ in our experiments) and decrease monotonically as $L$ increases. This trend aligns with the operator-spreading picture in \S~\ref{sec:analysis}: as $L$ grows, the backward lightcone of each local term expands, increasing the number of contributing components and enhancing cancellations, which reduces the typical curvature scale. Notably, across the entire depth range in Fig.~\ref{fig:depth}, the local objective retains markedly larger Hessian-entry fluctuations than the global objective, supporting the conclusion that locality delays curvature concentration at moderate depths. This depth dependence is consistent with the growth factor $V_G(k+2rL)$ appearing in~\eqref{eq:local_bound}.

\begin{figure}[htbp]
    \centering
    \includegraphics[width=0.45\textwidth]{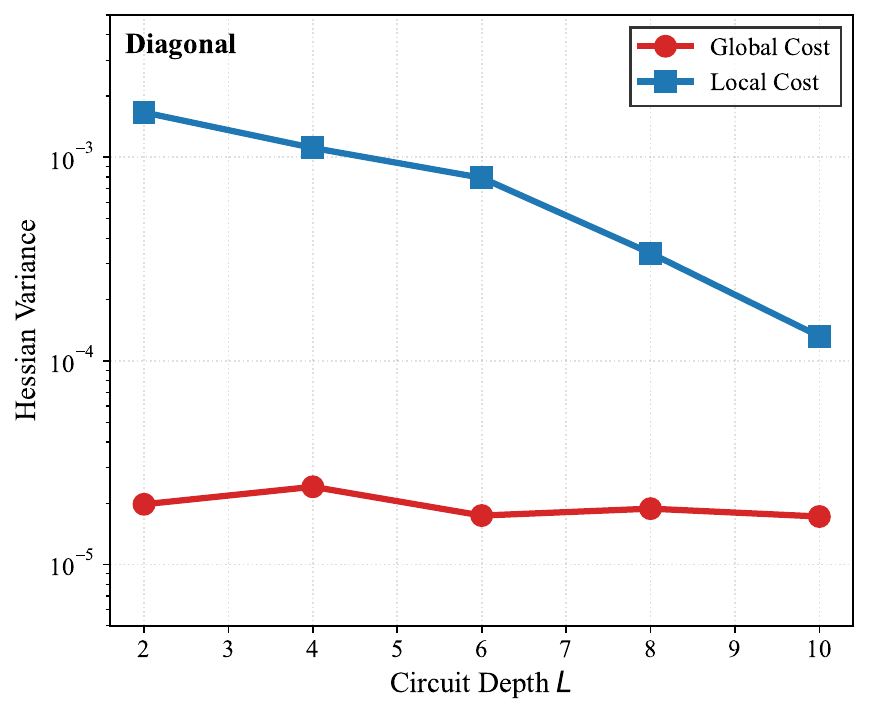}~~
    \includegraphics[width=0.45\textwidth]{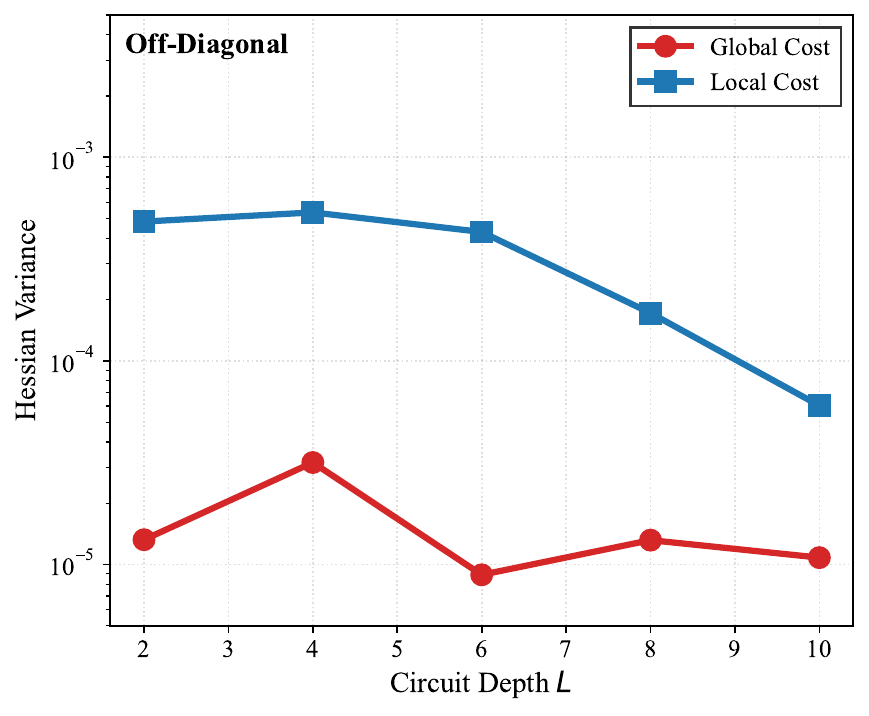}~~
    \caption{\textbf{Depth dependence of Hessian-entry variance.}
$\Var_\rho(H_{jj})$ (left) and $\Var_\rho(H_{jk})$ (right) versus depth $L$ at $n=16$ for $C_{\mathrm{global}}$ (red) and $C_{\mathrm{local}}$ (blue).}
    \label{fig:depth}
\end{figure}

\subsubsection{Robustness against measurement shot noise}

We assess finite-shot effects on Hessian-entry estimation at random initialization using the finite-shift estimator
$\widehat H_{jk}(\bth)=\sum_{s\in\mathcal S_{jk}} w_s\widehat C(\bth\oplus s)$, where each $\widehat C(\cdot)$ is estimated from $N$ shots. Corollary~\ref{cor:sample_complexity} implies $\Var_{\mathrm{sh}}(\widehat H_{jk}\mid\bth)=O(N^{-1})$, and averaging over $\bth\sim\rho$ yields
\[
\Var_{\rho,\mathrm{sh}}(\widehat H_{jk})
=\Var_\rho(H_{jk})+\E_{\bth\sim\rho}\left[\Var_{\mathrm{sh}}(\widehat H_{jk}\mid\bth)\right].
\]
Accordingly, we fit the two-term model $\Var_{\rho,\mathrm{sh}}(\widehat H_{jk}) \approx a+b/N$, where $a$ estimates the intrinsic initialization variance up to shift-rule constants.

In Fig.~\ref{fig:shot_noise}(a), the data exhibit an $N$-independent floor and a $1/N$ decay, consistent with this model. Figure~\ref{fig:shot_noise}(b) reports an \emph{absolute-tolerance} requirement
$N(\varepsilon):=\min\{N:\sqrt{\Var_{\rho,\mathrm{sh}}(\widehat H_{jk})}\le \varepsilon\}$ with $\varepsilon=0.05$; once $a\ll\varepsilon^2$, this criterion becomes shot-noise limited and $N(\varepsilon)\approx b/\varepsilon^2$, so the curves can saturate (or decrease) as $n$ increases. The scaling relevant to \emph{resolvability at initialization} is instead captured by $a$ (equivalently, $\Var_\rho(H_{jk})$ up to shift-rule constants) and the resolution criterion in Definition~\ref{def:resolution} (Corollary~\ref{cor:resolution_cost}), which ties the required shots to $\Var_\rho(H_{jk})^{-1}$.

\begin{figure}[t]
    \centering
    \includegraphics[width=0.45\textwidth]{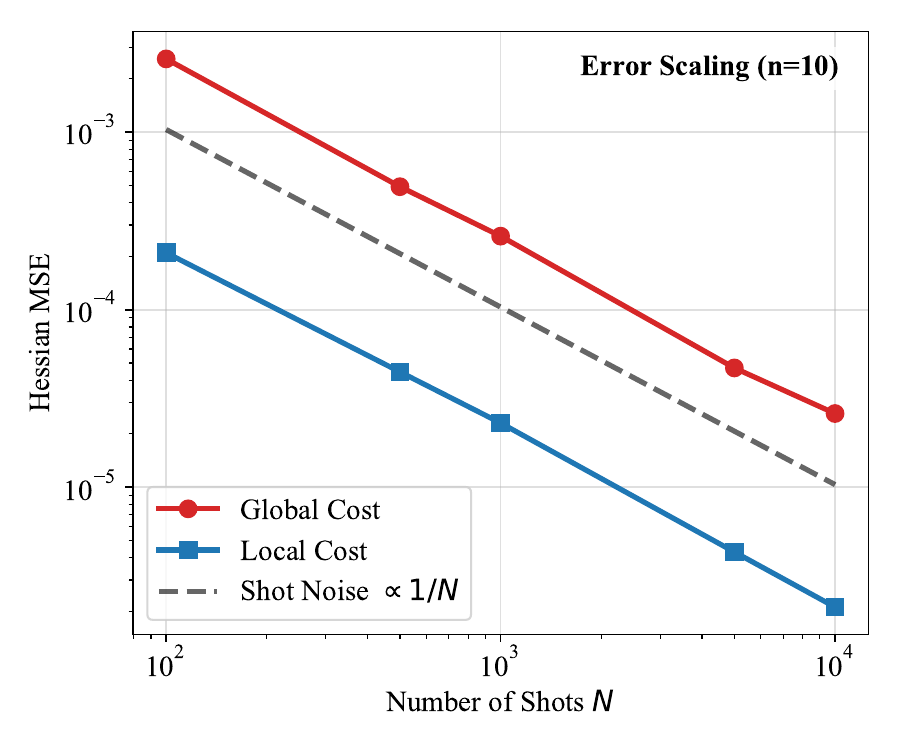}~~
    \includegraphics[width=0.45\textwidth]{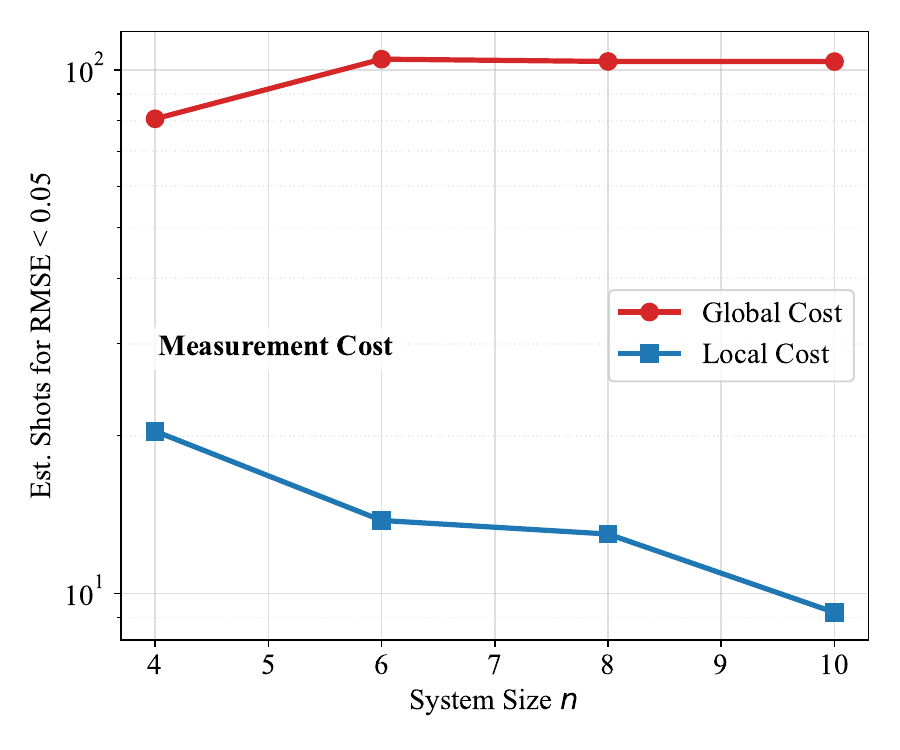}~~
    \caption{\textbf{Finite-shot effects on Hessian-entry estimation.}
(Left) $\Var_{\rho,\mathrm{sh}}(\widehat H_{jk})$ versus shots $N$. (Right) Absolute-tolerance requirement $N(\varepsilon)$ with $\varepsilon=0.05$ versus $n$.}
    \label{fig:shot_noise}
\end{figure}

\subsection{Spectral structure}
\label{sec:spectrum}

While our theory controls entry-wise curvature fluctuations, curvature-aware procedures also depend on how curvature
is distributed across parameter-space directions. We therefore report complementary spectral diagnostics based on the
empirical eigenvalue distribution of the Hessian at random initialization. These results are presented as numerical
diagnostics and are not interpreted as conditioning guarantees.

\paragraph{Empirical eigenspectrum}
Fig.~\ref{fig:spectrum} shows histogram estimates (log-density scale) of Hessian eigenvalues for global and term-wise
local averaged objectives at depth $L=4$ and increasing system size. As $n$ grows, the spectrum associated with the global
objective becomes increasingly concentrated near the origin, whereas the spectrum for the local objective remains
comparatively broader over the same range. This qualitative contrast mirrors the entry-wise variance behavior:
global objectives exhibit a rapidly shrinking initialization-scale curvature distribution, while locality delays such
concentration for term-wise averaged costs.

\begin{figure}[htbp]
    \centering
    \includegraphics[width=0.45\textwidth]{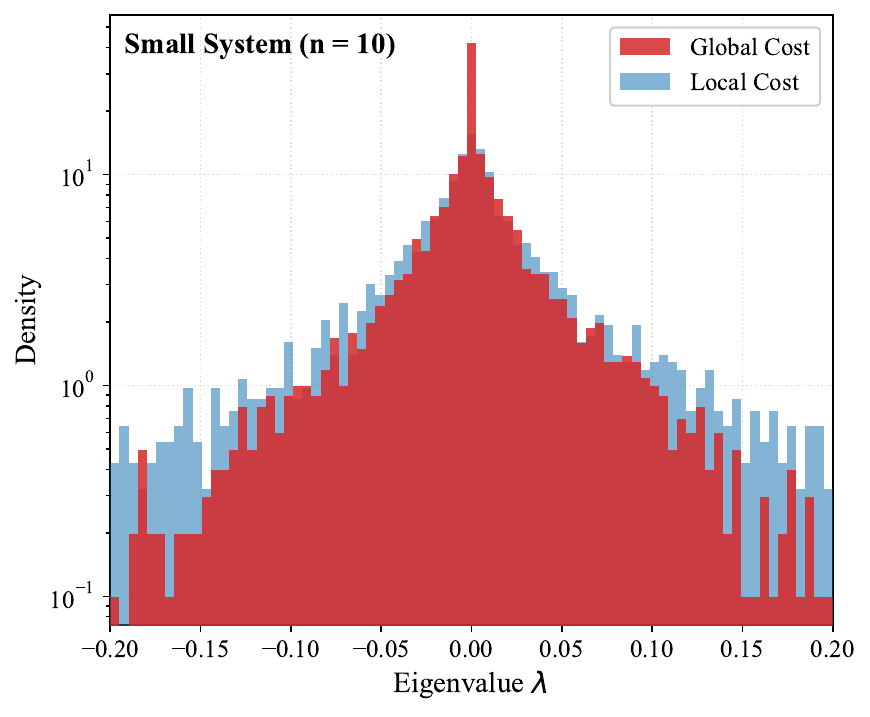}~~
    \includegraphics[width=0.45\textwidth]{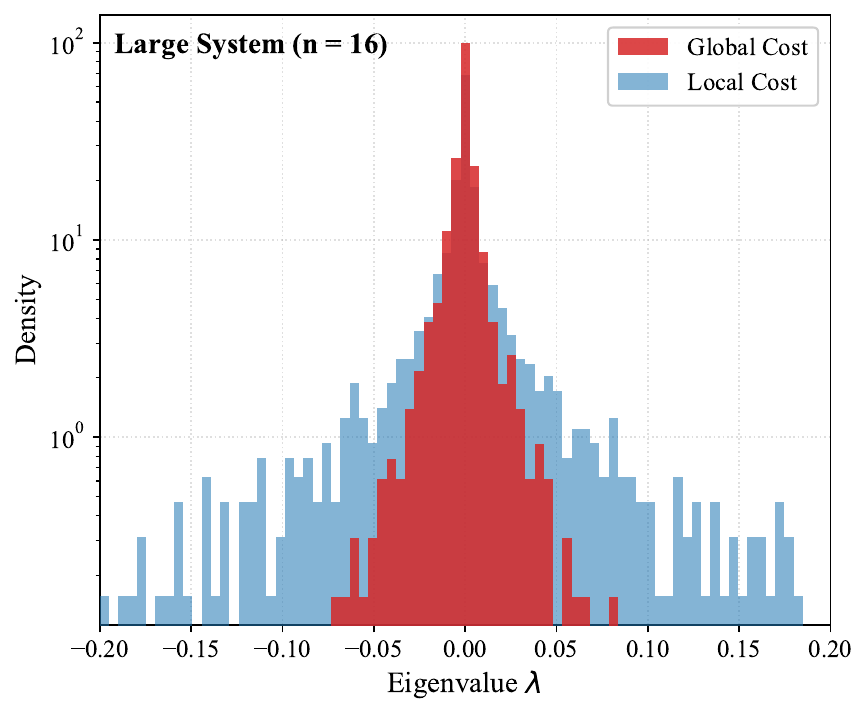}~~
    \caption{\textbf{Empirical Hessian eigenspectra at depth $L=4$.}
    Histogram estimates (log-density scale) of Hessian eigenvalues for $C_{\mathrm{global}}$ (red) and
    $C_{\mathrm{local}}$ (blue) at $n=10$ (left) and $n=16$ (right).}
    \label{fig:spectrum}
\end{figure}

\paragraph{Scalar spectral summaries}
To complement the histograms, we summarize the spectrum using two scalar metrics. The first is the root-mean-square (RMS)
eigenvalue $\lambda_{\mathrm{RMS}} := \sqrt{\frac{1}{M}\sum_{i=1}^{M}\lambda_i^2}$, which captures an overall curvature
magnitude scale. The second is a thresholded near-zero eigenvalue fraction, defined for a prescribed numerical tolerance
$\varepsilon>0$ as
$
\mathrm{Deg}_{\varepsilon} := \frac{1}{M}\#\{i:\ |\lambda_i|<\varepsilon\}.
$
In Fig.~\ref{fig:spectral_metrics} we report $\mathrm{Deg}_{\varepsilon}$ with $\varepsilon=10^{-4}$. We emphasize that
$\mathrm{Deg}_{\varepsilon}$ is resolution-dependent and should not be interpreted as a definitive statement about exact
rank deficiency.

Fig.~\ref{fig:spectral_metrics} (Left) shows that $\lambda_{\mathrm{RMS}}$ decreases with $n$ for both objectives, consistent
with an overall contraction of typical curvature magnitudes at initialization. Fig.~\ref{fig:spectral_metrics} (Right) shows
that, at the same numerical tolerance, the global objective exhibits a substantially larger near-zero fraction than the local
averaged objective. Interpreted at fixed resolution, this indicates more directions that appear numerically flat at initialization
for the global objective than for the term-wise averaged local objective.

\begin{figure}[t]
    \centering
    \includegraphics[width=0.45\textwidth]{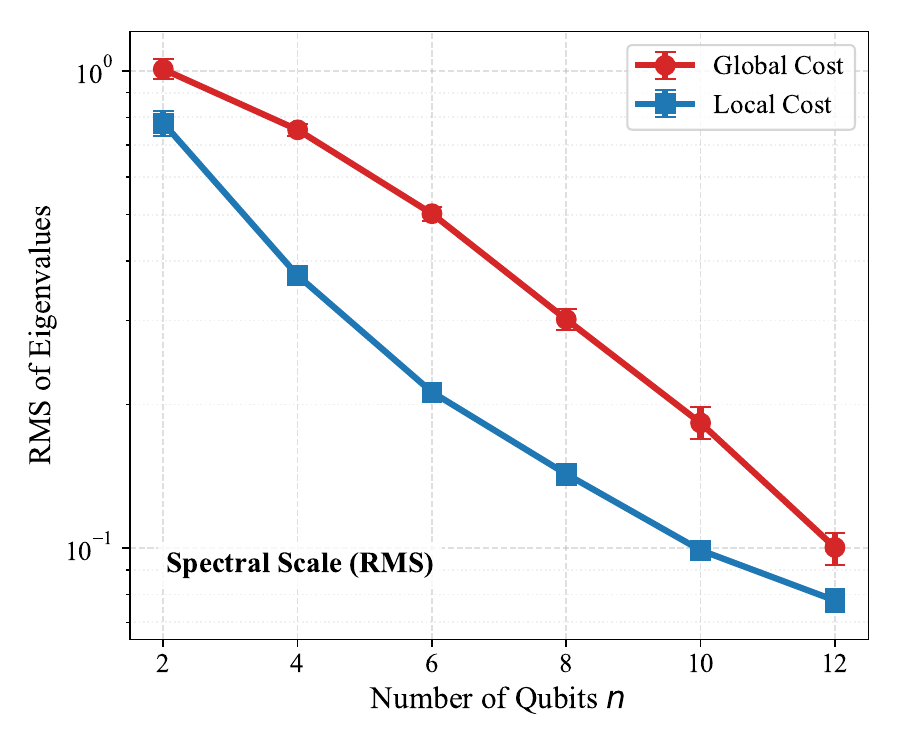}~~
    \includegraphics[width=0.45\textwidth]{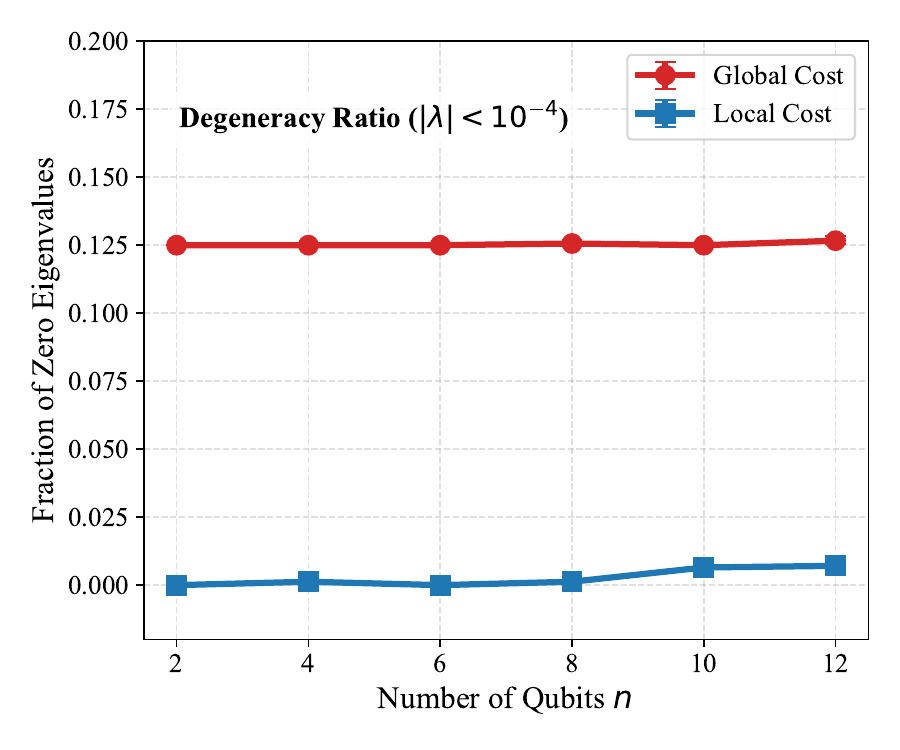}~~
    \caption{\textbf{Quantitative diagnostics of the Hessian spectrum at initialization.}
    (Left) RMS eigenvalue versus $n$ for the global cost (red circles) and local averaged cost (blue squares).
    (Right) Thresholded near-zero eigenvalue fraction with $\varepsilon=10^{-4}$.}
    \label{fig:spectral_metrics}
\end{figure}


\subsection{Finite-shot behavior near initialization}
\label{sec:algorithmic_implications}

The variance scalings in \S~\ref{sec:analysis} suggest that, under finite-shot sampling, global and term-wise local objectives can exhibit markedly different signal-to-noise behavior near random initialization. To illustrate this effect, we report matched optimization trajectories for a global objective and a term-wise local averaged objective under the same measurement protocol.

We fix $(n,L)=(18,8)$ and run stochastic gradient descent (SGD)~\cite{amari1993backpropagation} and the quantum natural gradient (QNG)~\cite{stokes2020quantum}. QNG is included only as a geometry-aware first-order baseline; the goal is not an optimizer comparison, but to probe how finite sampling affects observable progress near initialization.

Fig.~\ref{fig:optimization} reports the mean cost trajectories with $\pm1$ standard deviation bands over random initializations. For the global objective $C_{\mathrm{global}}$ (red), both methods exhibit trajectories that are nearly indistinguishable across seeds and show little change over the plotted horizon at the chosen sampling level. However, the term-wise local averaged objective $C_{\mathrm{local}}$ (blue) decreases consistently and displays visible seed-to-seed variability under the same protocol. These trends are consistent with \S~\ref{sec:analysis}: for global objectives, $\Var_\rho(H_{jk})$ decays exponentially in $n$, so at a fixed shot level the corresponding entry-wise fluctuations become difficult to resolve; for term-wise local objectives, the variance bounds are polynomial, so the same sampling protocol can still reveal measurable variation across initializations and iterations.

\begin{figure}[htbp]
    \centering
    \includegraphics[width=0.45\textwidth]{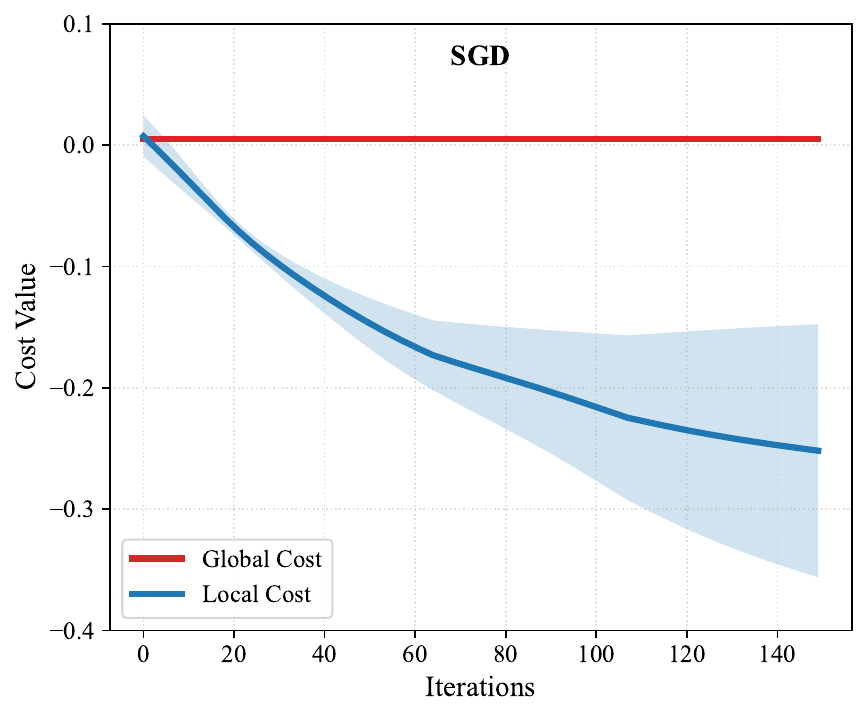}~~
    \includegraphics[width=0.45\textwidth]{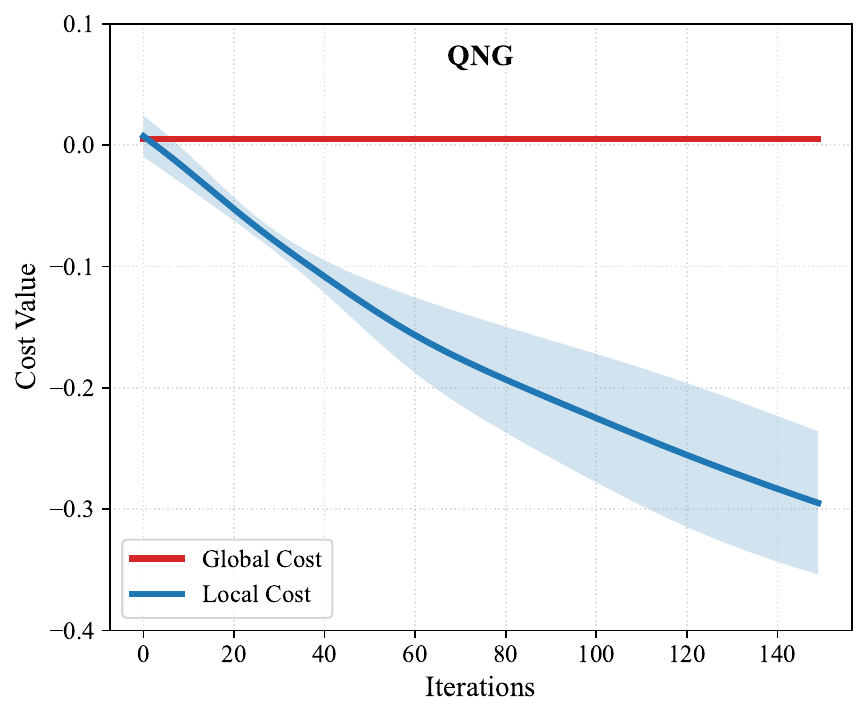}~~
    \caption{\textbf{Finite-shot optimization trajectories at $(n,L)=(18,8)$.}
Mean cost versus iteration for $C_{\mathrm{global}}$ (red) and $C_{\mathrm{local}}$ (blue): SGD (left) and QNG (right). Shading: $\pm 1$ standard deviation over random initializations.}
    \label{fig:optimization}
\end{figure}

\section{Conclusion}
\label{sec:conclusion}

In this study, we quantified the scaling of entry-wise Hessian variances in variational quantum algorithms. Using exact second-order parameter-shift identities, we expressed each Hessian entry as a constant-size linear combination of shifted objective evaluations, thereby yielding a finite covariance--quadratic representation for $\Var_\rho(H_{jk})$. This framework leads to two distinct scaling behaviors for the \emph{resolvability} of Hessian entries under finite-shot sampling. For global objectives satisfying a standard second-moment concentration condition, $\Var_\rho(H_{jk})$ decays exponentially in the number of qubits, and the corresponding resolution cost for Hessian-entry estimation grows exponentially with $n$. For term-wise $k$-local objectives in bounded-depth circuits, we obtain polynomial variance bounds governed by backward-lightcone growth on the interaction graph, which implies polynomial resolution cost whenever the lightcone remains sub-extensive. These conclusions are supported by numerical experiments across system sizes, circuit depths, and graph families.

Several avenues remain for future research. On the theory side, identifying conditions under which entry-wise resolvability correlates with informative spectral structure (e.g., conditioning or extremal eigenvalues) would clarify implications for Newton-type and preconditioned methods. On the practical side, extending the variance-based analysis to realistic noise models, including device noise and readout errors, is important for NISQ-relevant settings. More broadly, our results on derivative scaling inform higher-order quantum workflows in scientific computing, including quantum multiscale modeling~\cite{braun2025higher, chen2022qm, deiml2025quantum} and quantum machine learning~\cite{febrianto2025quantum, liu2025beyond}. Understanding resolvability limits in these contexts will be important for designing robust algorithms that rely on higher-order information.

\appendix

\section{Proofs}

\subsection{Covariance--quadratic representations from parameter shifts}
\label{app:cov_quad_ps}

\subsubsection{Proof of Lemma~\ref{lem:var_rep}}
\label{app:proof_var_rep}

\begin{proof}
All expectations, variances, and covariances are with respect to the initialization $\bth\sim\rho$.
Under Assumption~\ref{ass:obs}, $|C(\bth)|\le \|O\|\le 1$, so all second moments exist. By the diagonal second-order parameter-shift identity \eqref{eq:psr2-diag}, define
\[
C_{+}:=C(\bth+\pi\ej),\qquad C_{0}:=C(\bth),\qquad C_{-}:=C(\bth-\pi\ej).
\]
Then
\[
H_{jj}
=\frac14\big(C_{+}-2C_{0}+C_{-}\big)
=\sum_{\alpha\in\{+,0,-\}} w_\alpha\, C_\alpha,
\qquad
(w_+,w_0,w_-)=\tfrac14(1,-2,1).
\]
Using $\Var(\sum_\alpha w_\alpha C_\alpha)=\sum_{\alpha,\beta}w_\alpha w_\beta\,\Cov(C_\alpha,C_\beta)$ yields
\[
\Var(H_{jj})
=\sum_{\alpha,\beta\in\{+,0,-\}} w_\alpha w_\beta\,\Cov(C_\alpha,C_\beta)
= w^\top \Sigma w,
\]
where $\Sigma_{\alpha\beta}:=\Cov(C_\alpha,C_\beta)$. This proves \eqref{eq:var_Hjj_quad}. Expanding the quadratic form gives
\begin{align*}
\Var(H_{jj})
=\frac{1}{16}\Big(
&\Var(C_{+})+\Var(C_{-})+4\Var(C_{0})
+2\Cov(C_{+},C_{-})  \\
&-4\Cov(C_{+},C_{0})-4\Cov(C_{-},C_{0})
\Big),
\end{align*}
an explicit linear combination of variances and covariances among $\{C_+,C_0,C_-\}$, as stated in Lemma~\ref{lem:var_rep}.
\end{proof}

\subsubsection{Off-diagonal entries}
\label{app:offdiag_quad}

\begin{lemma}[Covariance--quadratic form for off-diagonal Hessian entries]
\label{lem:offdiag_quad}
Let $j\neq k$ and define the four shifted costs
\[
C_{\sigma\tau}:=C\Big(\bth+\sigma\tfrac{\pi}{2}\ej+\tau\tfrac{\pi}{2}\ek\Big),
\qquad \sigma,\tau\in\{+,-\}.
\]
Then \eqref{eq:psr2-off} implies
\[
H_{jk}=\frac14\Big(C_{++}-C_{+-}-C_{-+}+C_{--}\Big).
\]
Let $c:=(C_{++},C_{+-},C_{-+},C_{--})^\top$ and $w:=(1,-1,-1,1)^\top$, and set $\Sigma:=\Cov(c)\in\mathbb{R}^{4\times 4}$.
Then $\Var(H_{jk})=\frac1{16} w^\top \Sigma w$. Equivalently,
\begin{equation}
\label{eq:offdiag_var_16}
\Var(H_{jk})=\frac1{16}\sum_{\sigma,\tau\in\{+,-\}}\sum_{\sigma',\tau'\in\{+,-\}}
w_{\sigma\tau}\,w_{\sigma'\tau'}\Cov(C_{\sigma\tau},C_{\sigma'\tau'}),
\end{equation}
where $w_{++}=w_{--}=1$ and $w_{+-}=w_{-+}=-1$.
\end{lemma}

\begin{proof}
The mixed second-order parameter-shift identity \eqref{eq:psr2-off} gives the stated linear combination of the four shifted costs.
For any random vector $c$ and deterministic $a$, $\Var(a^\top c)=a^\top \Cov(c)\,a$. Applying this with $a=\frac14 s$ and expanding the quadratic form gives \eqref{eq:offdiag_var_16}.
\end{proof}

\begin{remark}
\label{rem:finite_shift_covquad}
More generally, any differentiation scheme that represents a derivative quantity as a finite linear combination of shifted costs,
\[
X(\bth)=\sum_{\ell=1}^m w_\ell C(\bth\oplus s_\ell), \qquad m<\infty,
\]
admits the covariance--quadratic representation
\[
\Var_\rho(X)=\sum_{\ell,\ell'=1}^m w_\ell w_{\ell'}\Cov_\rho\big(C(\bth\oplus s_\ell),C(\bth\oplus s_{\ell'})\big).
\]
In our setting, the relevant shift sets have constant size $m=O(1)$ determined only by the chosen parameter-shift rule. This is the general form of \eqref{eq:var_quadratic_form_intro}.
\end{remark}

\subsection{Covariance bounds under bounded spreading}
\label{app:cov_bounds}

For an observable $O$, define its Heisenberg evolution under the circuit by $O(\bth) := U^{\dagger}(\bth)OU(\bth)$.
Under Assumptions~\ref{ass:gen}--\ref{ass:depth}, there exists a backward-lightcone radius $R_L=O(rL)$ such that conjugation expands the support of any $k$-local operator to at most an $R_L$-neighborhood on the interaction graph $G$.

\begin{lemma}[Finite-range independence and a uniform covariance bound]
\label{lem:covdecay}
Under Assumptions~\ref{ass:gen}, \ref{ass:depth}, and \ref{ass:obs}, and under the initialization model in
\S~\ref{sec:sub:assumptions}, let $A$ and $B$ be Hermitian operators with
$\|A\|\le 1$ and $\|B\|\le 1$. Let $\Lambda_A:=\supp(A)$ and $\Lambda_B:=\supp(B)$, and define the random variables
\[
X(\bth):=\langle 0|A(\bth)|0\rangle,
\qquad
Y(\bth):=\langle 0|B(\bth)|0\rangle,
\]
where $A(\bth)=U(\bth)^\dagger A\,U(\bth)$ and $B(\bth)=U(\bth)^\dagger B\,U(\bth)$.
Let $d:=\dist_G(\Lambda_A,\Lambda_B)$ be the graph distance between the supports. Then:
\begin{enumerate}
\item[\textup{(i)}] If $d>2R_{L}$, then $X(\bth)$ and $Y(\bth)$ are independent and hence
\begin{equation}
\label{eq:cov_zero}
\Cov_\rho(X,Y)=0.
\end{equation}
\item[\textup{(ii)}] For all $d$, one has
\begin{equation}
\label{eq:cov_uniform}
|\Cov_\rho(X,Y)|\le 1.
\end{equation}
\end{enumerate}
Consequently,
\begin{equation}
\label{eq:cov_indicator}
|\Cov_\rho(X,Y)| \le \mathbf{1}\{d\le 2R_{L}\}.
\end{equation}
\end{lemma}

\begin{proof}
By Assumptions~\ref{ass:gen}--\ref{ass:depth}, $A(\bth)$ is supported within the $R_L$-neighborhood of $\Lambda_A$ and
$B(\bth)$ is supported within the $R_L$-neighborhood of $\Lambda_B$. Equivalently, $X(\bth)$ depends only on the subset
of parameters associated with gates whose qubit supports intersect the backward lightcone of $\Lambda_A$, and likewise
$Y(\bth)$ depends only on parameters within the backward lightcone of $\Lambda_B$.

If $d>2R_L$, these two backward lightcones are disjoint, hence the corresponding parameter subsets are disjoint.
Under the initialization model in \S~\ref{sec:sub:assumptions}, parameters on disjoint subsets are independent, so
$X(\bth)$ and $Y(\bth)$ are independent, proving \eqref{eq:cov_zero}.

For \eqref{eq:cov_uniform}, unitary conjugation preserves operator norm, so $\|A(\bth)\|=\|A\|\le 1$ and
$\|B(\bth)\|=\|B\|\le 1$. Therefore $|X(\bth)|\le 1$ and $|Y(\bth)|\le 1$, implying
$\Var_\rho(X)\le 1$ and $\Var_\rho(Y)\le 1$. Cauchy--Schwarz gives
$|\Cov_\rho(X,Y)|\le \sqrt{\Var_\rho(X)\Var_\rho(Y)}\le 1$, proving \eqref{eq:cov_uniform}.
Finally, \eqref{eq:cov_indicator} is the immediate combination of (i) and (ii).
\end{proof}

\begin{remark}
Lemma~\ref{lem:covdecay} uses only bounded operator spreading and independence of initialization parameters.
It yields a hard-cutoff correlation structure: beyond distance $2R_L$ the covariance vanishes, while for $d\le 2R_{L}$
we only use the trivial bound $|\Cov_\rho(X,Y)|\le 1$. No nontrivial decay inside the lightcone is claimed.
\end{remark}

\subsection{A dependency-graph variance bound}
\label{app:depgraph}

\begin{lemma}[Variance bound via a dependency graph]
\label{lem:depgraph_var}
Let $\{X_i\}_{i=1}^n$ be real-valued random variables with finite second moments, and let
$G_{\rm dep}=(\{1,\ldots,n\},E)$ be a dependency graph in the following sense:
for any disjoint index sets $I,J\subseteq\{1,\ldots,n\}$ with no edges between $I$ and $J$,
the families $\{X_i\}_{i\in I}$ and $\{X_j\}_{j\in J}$ are independent.
Let $\Delta_{\rm dep}$ be the maximum degree of $G_{\rm dep}$.
Then
\begin{equation}
\label{eq:depgraph_var_bound}
\Var\!\Big(\sum_{i=1}^n X_i\Big)
\;\le\; (\Delta_{\rm dep}+1)\sum_{i=1}^n \Var(X_i).
\end{equation}
Consequently, for $Z:=\frac1n\sum_{i=1}^n X_i$,
\begin{equation}
\label{eq:depgraph_var_bound_avg}
\Var(Z)\;\le\;\frac{\Delta_{\rm dep}+1}{n^2}\sum_{i=1}^n \Var(X_i).
\end{equation}
\end{lemma}

\begin{proof}
Expand the variance:
\[
\Var\!\Big(\sum_{i=1}^n X_i\Big)
=\sum_{i=1}^n \Var(X_i)+2\sum_{1\le i<j\le n}\Cov(X_i,X_j).
\]
If $\{i,j\}\notin E$, then there is no edge between $\{i\}$ and $\{j\}$, hence $X_i$ and $X_j$ are independent
by the dependency-graph property, so $\Cov(X_i,X_j)=0$.
Therefore,
\[
\Var\!\Big(\sum_{i=1}^n X_i\Big)
=\sum_{i=1}^n \Var(X_i)+2\sum_{\{i,j\}\in E,\,i<j}\Cov(X_i,X_j).
\]
Taking absolute values and applying Cauchy--Schwarz gives
\[
\big|\Cov(X_i,X_j)\big|\le \sqrt{\Var(X_i)\Var(X_j)} \le \tfrac12\big(\Var(X_i)+\Var(X_j)\big).
\]
Hence,
\begin{align*}
\Var\!\Big(\sum_{i=1}^n X_i\Big)
&\le \sum_{i=1}^n \Var(X_i)+2\sum_{\{i,j\}\in E,\,i<j} \tfrac12\big(\Var(X_i)+\Var(X_j)\big)\\
&= \sum_{i=1}^n \Var(X_i)+\sum_{\{i,j\}\in E}\big(\Var(X_i)+\Var(X_j)\big).
\end{align*}
Each $\Var(X_i)$ appears in the edge-sum exactly $\deg(i)$ times, hence
\[
\sum_{\{i,j\}\in E}\big(\Var(X_i)+\Var(X_j)\big)=\sum_{i=1}^n \deg(i)\,\Var(X_i)
\le \Delta_{\rm dep}\sum_{i=1}^n \Var(X_i).
\]
Combining these yields~\eqref{eq:depgraph_var_bound}, and~\eqref{eq:depgraph_var_bound_avg} follows by dividing by $n^2$.
\end{proof}

\subsection{Proof of Theorem~\ref{thm:scaling-explicit}}
\label{app:proof_scaling_explicit}

\begin{proof}
All expectations, variances, and covariances are with respect to the initialization law $\bth\sim\rho$; write
$\Var\equiv\Var_\rho$ and $\Cov\equiv\Cov_\rho$.

\paragraph{Step 1: Global objectives}
Assume $|\supp(O)|=\Theta(n)$ and Assumption~\ref{ass:GlobalDesign} holds. By shift-invariance of $\rho$ on the torus
(see Lemma~\ref{lem:shift_uniform}), for every fixed shift $s$ in a finite shift set,
\[
\Var\!\big(C(\bth\oplus s)\big)=\Var\!\big(C(\bth)\big)\le c_{\rm g}(r,L)\,e^{-\alpha(r,L)n}.
\]
Then for any $s,s'$ in the same finite set, Cauchy--Schwarz yields
\[
\big|\Cov(C(\bth\oplus s),C(\bth\oplus s'))\big|
\le \sqrt{\Var(C(\bth\oplus s))\Var(C(\bth\oplus s'))}
\le c_{\rm g}(r,L)\,e^{-\alpha(r,L)n}.
\]
Substituting this uniform bound into \eqref{eq:psr2-off} and using that $|\mathcal S_{ab}|=O(1)$ gives
\[
\Var(H_{jk})\le \tilde c(r,L)\,e^{-\tilde\alpha(r,L)n},
\]
for constants depending only on $(r,L)$ and the chosen shift rule (via $\{a_s\}$ and $|\mathcal S_{ab}|$). This proves
\eqref{eq:global_var_bound}.

\paragraph{Step 2: Term-wise local averaged objectives (polynomial scaling)}
Assume the objective is of the term-wise local averaged form
\[
C(\bth)=\frac1n\sum_{v=1}^n \langle 0|\,O_v(\bth)\,|0\rangle,
\qquad
O_v(\bth):=U(\bth)^\dagger O_v U(\bth),
\qquad \|O_v\|\le 1,
\]
where each $O_v$ is supported on at most $k$ qubits (with $k$ independent of $n$). Fix any shift $s$.
Define the local random variables
\[
X_v^{(s)}(\bth):=\langle 0|\,U(\bth\oplus s)^\dagger O_v U(\bth\oplus s)\,|0\rangle,
\qquad
C(\bth\oplus s)=\frac1n\sum_{v=1}^n X_v^{(s)}(\bth).
\]
Shifting parameters does not change the circuit geometry, hence the backward-lightcone radius remains $O(rL)$ for every $s$.
By bounded spreading (Assumptions~\ref{ass:gen}--\ref{ass:depth}), $X_v^{(s)}$ depends only on parameters in the
backward lightcone of $\supp(O_v)$, whose qubit support is contained in a $(k+rL)$-neighborhood. Therefore, if
$\dist_G(\supp(O_v),\supp(O_w))>k+2rL$, the parameter sets influencing $X_v^{(s)}$ and $X_w^{(s')}$ are disjoint,
and the two variables are independent (Lemma~\ref{lem:covdecay}).

Construct a dependency graph on vertices $\{1,\dots,n\}$ by connecting $v\sim w$ whenever
$\dist_G(\supp(O_v),\supp(O_w))\le k+2rL$. Its maximum degree is bounded by
$\Delta_{\rm dep}\le V_G(k+2rL)-1$. Since $|X_v^{(s)}|\le 1$, we have $\Var(X_v^{(s)})\le 1$.
Applying Lemma~\ref{lem:depgraph_var} to $\sum_v X_v^{(s)}$ yields the uniform bound (in $s$)
\begin{equation}
\label{eq:A4_var_shifted_cost}
\Var\!\big(C(\bth\oplus s)\big)
=\Var\!\Big(\frac1n\sum_{v=1}^n X_v^{(s)}\Big)
\le \frac{\Delta_{\rm dep}+1}{n^2}\sum_{v=1}^n \Var(X_v^{(s)})
\le \frac{V_G(k+2rL)}{n}.
\end{equation}
The same construction applies to cross-covariances between two shifts $s$ and $s'$: independence still holds whenever
$\dist_G(\supp(O_v),\supp(O_w))>k+2rL$, so the same dependency-graph argument gives
\begin{equation}
\label{eq:A4_cov_shifted_cost}
\big|\Cov(C(\bth\oplus s),C(\bth\oplus s'))\big|\le \frac{V_G(k+2rL)}{n},
\qquad \forall s,s' \text{ in a fixed finite shift set}.
\end{equation}

Finally, substitute \eqref{eq:A4_cov_shifted_cost} into the covariance--quadratic representation
\eqref{eq:psr2-off}. Since $|\mathcal S_{ab}|=O(1)$ and the coefficients $\{a_s\}$ are constants determined
by the shift rule, we obtain
\[
\Var(H_{jk})\le c_{\rm loc}(k,r,L)\,\frac{V_G(k+2rL)}{n},
\]
for some constant $c_{\rm loc}(k,r,L)>0$. This proves \eqref{eq:local_bound}. The lattice-growth specialization
\eqref{eq:local_bound_poly_growth} follows immediately from $V_G(m)=O(m^D)$.
\end{proof}

\section{Supplementary}

\subsection{Uniformity over finite shift sets}
\label{app:global_shift_uniform}

\begin{lemma}[Uniform concentration over finite shift sets]
\label{lem:shift_uniform}
Let $\mathcal S\subset\R^M$ be any finite set of shifts and interpret parameters on the torus
$\mathbb T^M=(\R/2\pi\Z)^M$.
Suppose the initialization distribution $\rho$ on $\mathbb T^M$ is shift-invariant with respect to $\mathcal S$, i.e.,
\begin{equation}
\label{eq:rho_shift_invariant}
\bth\sim\rho  \Longrightarrow  \bth\oplus s\sim\rho,\qquad \forall s\in\mathcal S,
\end{equation}
where $\oplus$ denotes addition modulo $2\pi$ componentwise.
Then Assumption~\ref{ass:GlobalDesign} implies
\begin{equation}
\label{eq:global_concentration_shifts}
\sup_{s\in \mathcal S}\Var_\rho\big[C(\bth\oplus s)\big]\le c_{\rm g}(r,L)e^{-\alpha(r,L)n}.
\end{equation}
\end{lemma}

\begin{proof}
Fix $s\in\mathcal S$. By \eqref{eq:rho_shift_invariant}, the random variables
$C(\bth)$ and $C(\bth\oplus s)$ have the same distribution under $\rho$, hence
$\Var_\rho[C(\bth\oplus s)]=\Var_\rho[C(\bth)]$. Taking the supremum over $s\in\mathcal S$
yields \eqref{eq:global_concentration_shifts}.
\end{proof}

If $\rho$ is the product of i.i.d.\ uniform measures on $[0,2\pi)$,
then \eqref{eq:rho_shift_invariant} holds for any finite shift set $\mathcal S$ arising in parameter-shift rules.

\subsection{A sufficient design-based route for global concentration}
\label{app:design_route}

\paragraph{How to verify Assumption~\ref{ass:GlobalDesign}}
Assumption~\ref{ass:GlobalDesign} can be supported (i) theoretically via design/ mixing conditions, and/or
(ii) empirically by Monte Carlo sampling $\bth^{(m)}\sim\rho$ and estimating $\Var_\rho[C(\bth)]$ from the sample variance.
By Lemma~\ref{lem:shift_uniform}, the same empirical test applies uniformly over the finite shift set used in
parameter-shift rules.

\begin{proposition}[A sufficient condition for Assumption~\ref{ass:GlobalDesign}]
\label{prop:2design_implies_global}
Let $d=2^n$. Assume the induced state ensemble
$\{|\psi(\bth)\rangle:=U(\bth)|0\rangle\}_{\bth\sim\rho}$
forms an exact complex projective $2$-design on $\mathbb{C}^d$.
Then for any Hermitian observable $O$,
\begin{equation}
\label{eq:haar_var_formula}
\Var_\rho\big[\langle\psi(\bth)|O|\psi(\bth)\rangle\big]
=\frac{\Tr(O^2)-\Tr(O)^2/d}{d(d+1)}.
\end{equation}
In particular, if $\|O\|\le 1$ then $\Tr(O^2)\le d$, hence
\begin{equation}
\label{eq:haar_var_bound}
\Var_\rho\big[C(\bth)\big]\le \frac{1}{d+1}=\mathcal{O}(2^{-n}),
\end{equation}
so Assumption~\ref{ass:GlobalDesign} holds with $\alpha=\log 2$ up to constants.
\end{proposition}

\begin{proof}
Let $|\psi\rangle$ be distributed according to an exact projective $2$-design on $\mathbb{C}^d$.
By the defining second-moment property of a projective $2$-design, the second moment operator matches Haar:
\begin{equation}
\label{eq:second_moment}
\mathbb{E}\big[|\psi\rangle\langle\psi|^{\otimes 2}\big]
=\frac{I+F}{d(d+1)},
\end{equation}
where $F$ is the swap operator on $\mathbb{C}^d\otimes\mathbb{C}^d$.
Then
\[
\mathbb{E}\big[\langle\psi|O|\psi\rangle\big]
=\Tr\Big(O\mathbb{E}[|\psi\rangle\langle\psi|]\Big)
=\frac{\Tr(O)}{d},
\]
using $\mathbb{E}[|\psi\rangle\langle\psi|]=I/d$.
For the second moment, use the swap trick:
\[
\langle\psi|O|\psi\rangle^2
=\Tr\big((O\otimes O)|\psi\rangle\langle\psi|^{\otimes 2}\big),
\]
hence by~\eqref{eq:second_moment},
\[
\mathbb{E}\big[\langle\psi|O|\psi\rangle^2\big]
=\Tr\Big((O\otimes O)\frac{I+F}{d(d+1)}\Big)
=\frac{\Tr(O)^2+\Tr(O^2)}{d(d+1)}.
\]
Subtracting the squared mean $(\Tr(O)/d)^2$ yields~\eqref{eq:haar_var_formula}.
If $\|O\|\le 1$, then $\Tr(O^2)\le d\|O\|^2\le d$, giving~\eqref{eq:haar_var_bound}.
\end{proof}

We do not claim that modest-depth hardware-efficient ans\"atze form $2$-designs; the proposition is included only as a
verifiable sufficient route to global concentration when design/mixing guarantees are available.

\subsection{From entry-wise second moments to matrix norm bounds}
\label{app:norm_bounds}

We provide a proof of Lemma~\ref{lem:entrywise-to-norms} stated in the main text.

\begin{proof}[Proof of Lemma~\ref{lem:entrywise-to-norms}]
Let $H\in\mathbb{R}^{M\times M}$ be a random matrix with finite second moments. By definition, $\|H\|_F^2=\sum_{j,k=1}^M H_{jk}^2$, and taking expectations gives
\[
\E\|H\|_F^2=\sum_{j,k=1}^M \E[H_{jk}^2]
=\sum_{j,k=1}^M\Big(\Var(H_{jk})+(\E[H_{jk}])^2\Big).
\]
Since $\|H\|_2\le \|H\|_F$ pointwise, we have $\E\|H\|_2^2 \le \E\|H\|_F^2$.
Moreover, Jensen's inequality yields
\[
\E\|H\|_2 \le \E\|H\|_F \le \sqrt{\E\|H\|_F^2}.
\]
If $\Var(H_{jk})\le v_n$ and $\big|\E [H_{jk}]\big|\le \mu_n$ for all $j,k$, then
\[
\E\|H\|_F^2
=\sum_{j,k=1}^M\Big(\Var(H_{jk})+\big(\E[H_{jk}]\big)^2\Big)
\le M^2(v_n+\mu_n^2),
\]
and the spectral norm bounds above give
\[
\E\|H\|_2^2 \le M^2(v_n+\mu_n^2),
\qquad
\E\|H\|_2 \le M\sqrt{v_n+\mu_n^2}.
\]
This completes the proof.
\end{proof}

\vspace{1cm}

\bibliographystyle{siamplain}
\bibliography{bib}

\end{document}